\documentclass[10pt]{article}

\usepackage{amsmath,amssymb,amsthm,color,epsfig,mathrsfs}
\usepackage{amsbsy,accents}

\usepackage[section]{placeins}

\newtheorem{theorem}{Theorem}

\newtheorem{corollary}[theorem]{Corollary}

\newtheorem{definition}[theorem]{Definition}

\newtheorem{lemma}[theorem]{Lemma}

\newtheorem{proposition}[theorem]{Proposition}

\newcounter{spslist}
\newenvironment{spslist}{
  \begin{list}
  {\arabic{spslist}.}
  {\usecounter{spslist}
  \setlength{\leftmargin}{\parindent}
  \setlength{\labelsep}{0.6em}
  \setlength{\labelwidth}{1em}
  \setlength{\topsep}{0.8ex}
  \setlength{\rightmargin}{0em}
  \setlength{\itemsep}{0.5ex}
  \setlength{\parsep}{0em}
  \setlength{\itemindent}{0em} }}
  {\end{list}}

\newcommand{\mat}[5]{ \renewcommand{\arraystretch}{#1}
                    \left[\! \begin{array}{cc}
                            #2 & #3 \\
                            #4 & #5 \end{array} \!\right] }

\newcounter{geqncount}
    {\refstepcounter{equation}%
     \setcounter{geqncount}{\value{equation}}%
     \setcounter{equation}{0}%
  }%
    {\setcounter{equation}{\value{geqncount}}}
\newcommand{\V}{\mathcal{V}}
\newcommand{\E}{\mathcal{E}}

\newcommand{\G}{\Gamma}
\newcommand{\Go}{\mathring{\Gamma}}
\newcommand{\Lo}{\mathring{\Lambda}}

\newcommand{\Ao}{\mathring{A}}

\newcommand{\ZZ}{\mathbb{Z}}
\newcommand{\RR}{\mathbb{R}}
\newcommand{\CC}{\mathbb{C}}
\newcommand{\TT}{\mathbb{T}}

\newcommand{\hDD}{h^\text{DD}}
\newcommand{\hDR}{h^\text{DR}}
\newcommand{\hRR}{h^\text{RR}}
\newcommand{\hRD}{h^\text{RD}}

\newcommand{\dhDD}{\dot h^\text{DD}}
\newcommand{\dhDR}{\dot h^\text{DR}}
\newcommand{\dhRR}{\dot h^\text{RR}}
\newcommand{\dhRD}{\dot h^\text{RD}}

\newcommand{\conn}{\mathrm{c}}
\newcommand{\cchi}{\text{\raisebox{1.5pt}{$\chi$}}}

\begin{document}

\bibliographystyle{plain}

\begin{center}
{\bf \Large  Reducible Fermi surface for multi-layer quantum graphs \\ \vspace{0.7ex} including stacked graphene}
\end{center}

\vspace{0.2ex}

\begin{center}
{\scshape \large Lee Fisher* \,and\, Wei Li** \,and\, Stephen P\hspace{-2pt}. Shipman**} \\
\vspace{1ex}
{\itshape Departments of Mathematics\\ University of California at Irvine* and Louisiana State University**}
\end{center}

\vspace{3ex}
\centerline{\parbox{0.9\textwidth}{
{\bf Abstract.}\
We construct two types of multi-layer quantum graphs (Schr\"odinger operators on metric graphs) for which the dispersion function of wave vector and energy is proved to be a polynomial in the dispersion function of the single layer.  This leads to the reducibility of the algebraic Fermi surface, at any energy, into several components.  Each component contributes a set of bands to the spectrum of the graph operator.  When the layers are graphene, AA-, AB-, and ABC-stacking are allowed within the same multi-layer structure.  Conical singularities (Dirac cones) characteristic of single-layer graphene break when multiple layers are coupled, except for special AA-stacking.
One of the tools we introduce is a surgery-type calculus for obtaining the dispersion function for a periodic quantum graph by gluing two graphs together.
}}

\vspace{3ex}
\noindent
\begin{mbox}
\end{mbox}

\vspace{3ex}

\noindent
\begin{mbox}
\end{mbox}
\vspace{3ex}

%%%%%%%%%%%%%%%%%%%%%%%%%%%%%%%%%%%%%%%%%%%%%%%%%
%%%%%%%%%%%%%%%%%%%%%%%%%%%%%%%%%%%%%%%%%%%%%%%%%
\section{Introduction} 

The Fermi surface of a $d$-periodic operator at an energy $\lambda$ is the set of wavevectors  $(k_1,\dots,k_d)$ admissible by the operator at that energy.  
It is the zero-set of the dispersion function $D(k_1,\dots,k_d,\lambda)$, and for periodic graph operators, this is a Laurent polynomial in the Floquet variables $(z_1,\dots,z_d)=(e^{ik_1},\dots,e^{ik_d})$.
When the dispersion function can be factored, for each fixed energy, as a product of two or more polynomials in $(z_1,\dots,z_d)$, each irreducible component contributes a sequence of spectral bands and gaps.  Reducibility is required for the existence of embedded eigenvalues engendered by a local defect \cite{KuchmentVainberg2000,KuchmentVainberg2006}, except for the anomalous situation when an eigenfunction has compact support.

Irreducibility of the Fermi surface for all but finitely many energies is known to occur in special cases:  for the 2D and 3D discrete Laplacian plus a periodic potential~\cite{Battig1988},\cite[Ch.~4]{GiesekerKnorrerTrubowitz1993},\cite[Theorem~2]{Battig1992}; for the continuous Laplacian plus a potential of the form $q(x_1)+q(x_2,x_3)$~\cite[Sec.~2]{BattigKnorrerTrubowitz1991}; and for discrete graph Laplacians with positive weights and more general graph operators, where the underlying graph is planar with two vertices per period~\cite{LiShipman2019a}.

Reducible Fermi surfaces are known to occur for certain multi-layer graph operators.  The simplest are constructed by coupling several identical copies of a discrete graph operator, using Hermitian coupling constants~\cite[\S2]{Shipman2014}; or by coupling two identical layers of a quantum graph by edges between corresponding vertices, where the potential $q_e(x)$ of the Schr\"odinger operator $-d^2/dx^2 + q_e(x)$ on each coupling edge $e$ is symmetric about the center of the edge~\cite[\S3]{Shipman2014}.  When the layers are not coupled symmetrically, a condition for reducibility is that the potentials $q_e(x)$ on all of the coupling edges must lie in the same ``asymmetry class"~\cite[Theorem~4]{Shipman2019}, which occurs when they possess identical asymmetry functions $a_{q_e}(\lambda)$ as defined in that article.  The mechanism at work is a decomposition of the discrete-graph reduction of the quantum graph at each energy, enabled by the common asymmetry function.  

\smallskip

The present work introduces two new classes of self-adjoint periodic multi-layer quantum graph operators with reducible Fermi surface.  For both types, several layers are connected by general graphs, as illustrated in Fig.\,\ref{fig:Layered}.
Reducibility results from a different mechanism from~\cite{Shipman2019}, namely that the dispersion function of the multi-layer graph turns out to be a polynomial function of a ``composite Floquet variable" $\zeta(z_1,\dots,z_d,\lambda)$,
\begin{equation*}
  D(z_1,\dots,z_d,\lambda) \;=\; P\big(\zeta(z_1,\dots,z_d,\lambda),\lambda\big).
\end{equation*}
The function $\zeta$ is a Laurent polynomial in $(z_1,\dots,z_d)$ with coefficients that are meromorphic in~$\lambda$; and $P$ is a polynomial in~$\zeta$ with coefficients that are meromorphic in~$\lambda$.  The components of the Fermi surface at energy~$\lambda$ are therefore of the form $\zeta(z_1,\dots,z_d;\lambda) = \rho(\lambda)$, where $\rho(\lambda)$ is one of the roots of $P(\zeta,\lambda)$ as a polynomial in~$\zeta$.  For multi-layer graphene with the potentials on the edges of the layers being isospectral, this reduces to the form
\begin{equation}
   \mu(\lambda) \;=\; \tilde G(k_1,k_2),
\end{equation}
as discussed in section~\ref{sec:graphene}.  This allows an easy numerical computation of the spectrum.
% $\lambda\in\RR$ is in the spectrum of the periodic multi-layer operator whenever it is in the range of the function~$\tilde G(k_1,k_2)$.

Here is a summary of properties of the two types of multi-layer quantum graphs, introduced in this work, that have reducible Fermi surface, with the number of components equal to the number of layers.  All potentials are electric; we do not treat magnetic potentials in this work.

\medskip
\noindent
{\bfseries Type-1 multilayer graphs.} (Fig.\,\ref{fig:Layered} left; and section\,\ref{sec:type1})

\begin{spslist}
  \item Each layer is separable---it breaks into an infinite array of finite pieces when a vertex and its shifts are removed (Fig.~\ref{fig:Separable}).
  \item The dispersion function $D(z,\lambda)$ of each layer is a polynomial function of a single function $\zeta(z,\lambda)$; for example, all layers may all have the same dispersion function.
  \item Several layers are connected at the vertices of separation by edges or by a more complicated graph.
    \item  An example is AB-stacked graphene, with arbitrary potentials on the three edges of a period of each layer.  Layers may be rotated $180^\circ$.  See section~\ref{subsec:AB}.
\end{spslist}

The proof of Theorem~\ref{thm:type1} on type 1 employs a calculus for obtaining the dispersion function of a periodic graph in terms of the dispersion functions of component graphs, where the component graphs are joined at the periodic shifts of a single vertex (Lemma~\ref{lemma:svj}).  It is a sort of periodic-graph version of surgery techniques for compact graphs, studied in~\cite{BerkolaikoKennedyKurasov2019a}.

\medskip
\noindent
{\bfseries Type-2 multilayer graphs.} (Fig.\,\ref{fig:Layered} right; and section\,\ref{sec:type2})

\begin{spslist}
  \item Each layer has the same underlying graph, which is bipartite with one vertex of each color per period.
  \item  The potentials on corresponding edges in different layers are isospectral in the Dirichlet sense.
  \item Several layers are connected along vertices of the same color by edges or more complicated graphs.
  \item An example is AA-stacked graphene, with arbitrary potentials on the three edges of a period.  Layers may be rotated $180^\circ$.  See section~\ref{subsec:AA}.
\end{spslist}

The proof of Theorem~\ref{thm:type2} on type 2 is linear-algebraic.  It is a generalization of AA-stacked bilayer graphene discussed in~\cite[\S6]{Shipman2019}, where it was noticed that the potentials of the connecting edges, miraculously, did not have to lie in the same asymmetry class in order to obtain reducibility of the Fermi surface.

\medskip
\noindent
{\bfseries Multi-layer graphene models.} (Section\,\ref{sec:graphene})
\smallskip

The quantum-graph model of single-layer graphene satisfies the properties of the individual layers of both type\,1 and type\,2:  (a) Being bipartite, each layer is separable at any vertex; and (b) isospectrality of the potentials on corresponding edges across layers implies that each layer is a polynomial in a common function~$\zeta(z,\lambda)$.
By applying the techniques of both types, one finds that very general stacking of graphene into several layers has reducible Fermi surface.  This includes AA-, AB-, ABC-, and mixed stacking, as illustrated in the figures of Section~\ref{sec:graphene}.  

There is a huge amount of literature on the properties of graphene and its variations---electronic, spectral, and other physical and mathematical properties---and fascinating applications.
The differences between single- and multiple-layer graphene are expounded in~\cite{Castro-NetGuineaPeres2009a,AbergelApalkovBerashevic2010a}, which offers much physical context.  The most important feature of two or more layers is a transition from conical singularities of the dispersion relation (linear band structure at Dirac points) for a single layer to nonconical singularities (quadratic band structure) for multiple layers, accompanied by spectral gaps; see~\cite{PartoensPeeters2006a,McCann2006a,CastroNovoselovMorozov2007a}, for example.  Refs. \cite{KimWalterMoreschini2013a,RozhkovSboychakovRakhmanov2017a,SboychakovRakhmanovRozhkov2015a} present some interesting work on opening gaps by twisting two layers relative to one another.

The present work contributes to the spectral properties of graph models of multi-layer graphene with electric potentials in these ways:  (1) The reducibility of the Fermi surface opens the possibility of constructing local defects in multi-layer graphene that would allow bound states within the radiation continuum (cf.~\cite{KuchmentVainberg2006}); (2) The theorems demonstrate the range of allowed potentials on the layers and the connecting edges in order to obtain a reducible Fermi surface, and (3) The theory offers insight into the effect of multiple layers on conical singularities, or Dirac cones.  Particularly, a condition for conical singularities to persist in AA-stacked graphene is given (section~\ref{subsec:conical}, Proposition~\ref{prop:AAconical}); this result is subsumed by~\cite[Theorem\,2.4]{BerkolaikoComech2018}, which exploits symmetries to obtain Dirac cones; here it arises through different calculations.

%A detailed account of multi-layer quantum-graph graphene alone would be an interesting work in its own right.  The emphasis of this work is on the mechanisms of reducibility of the Fermi surface, which applies to  classes of multi-layer graphs that are much broader than graphene, including arbitrary number of dimensions of periodicity.  

\begin{figure}[h]
\centerline{
\scalebox{0.46}{\includegraphics{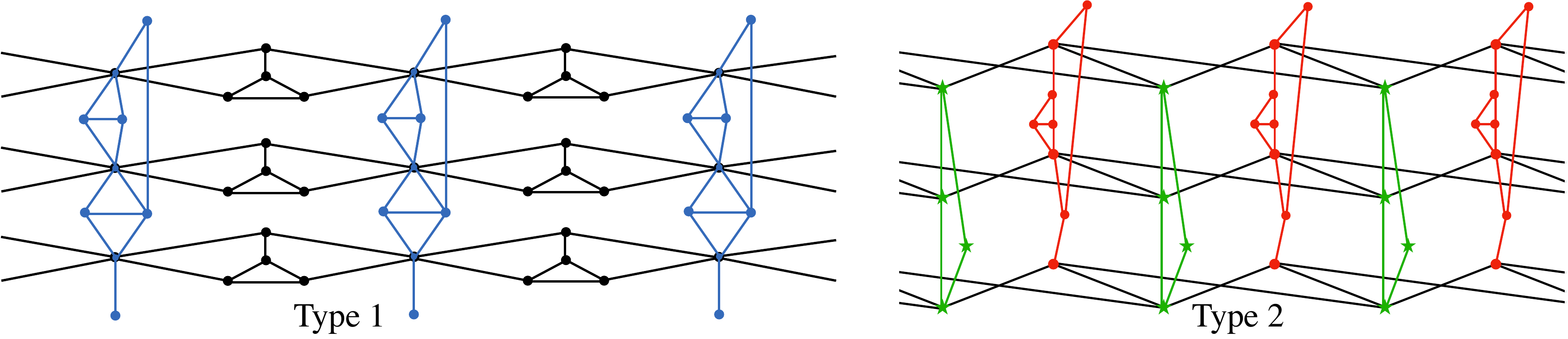}}
}
\caption{\small $1$-periodic examples of multi-layer graphs with reducible Fermi surface. Type~1 (left):  Each layer is a separable periodic graph whose dispersion function is a polynomial in a fixed Laurent polynomial $\zeta(z,\lambda)$.  The layers are connected at corresponding vertices of separation by the periodic translates of a finite (blue) graph.  Type~2 (right):  Each layer has the same underlying bipartite periodic graph with two vertices per period, and the potentials on corresponding (vertically displaced) edges have the same Dirichlet spectrum.  Vertically displaced green vertices are connected by periodic translates of a finite (green) graph; and similarly for the red vertices.}
\label{fig:Layered}
\end{figure}

%%%%%%%%%%%%%%%%%%%%%%%%%%%%%%%%%%%%%%%%%%%%%%%%%
%%%%%%%%%%%%%%%%%%%%%%%%%%%%%%%%%%%%%%%%%%%%%%%%%
\section{Periodic quantum graphs, Fermi surface, reducibility}\label{sec:background}

This section lays down the notation for periodic quantum graphs and the Fermi surface and so on.  Lemma~\ref{lemma:svj} introduces a calculus for joining two periodic graphs by the ``single-vertex join", defined in Definition~\ref{def:svj}, and serves as the basic tool for proving reducibility of multi-layer graphs of type\,1.

%%%%%%%%%%%%%%%%%%%%%%%%%%%%%%%%%%%%%%%%%%%%%%%%%
\subsection{Quantum graphs and notation}\label{subsec:notation}

We give a brief account of periodic quantum graphs that is sufficient for the analysis in this paper.  The notation specific to periodic graphs essentially follows~\cite[\S3.1-3.2]{Shipman2019}; and the standard text~\cite{BerkolaikoKuchment2013} provides a more general exposition of quantum graphs.

The structure of a periodic quantum graph begins with an underlying graph $\G$, with vertex set $\V(\G)$ and edge set $\E(\G)$, with a free shift action of $\ZZ^d$, denoted by $x\mapsto gx$ for $x\in\G$ and $g=(g_1,\dots,g_d)\in\ZZ^d$.  We assume that a fundamental domain of the action (a period of the graph) has finitely many vertices and edges.  $\G$ becomes a metric graph when each edge $e$ is coordinatized by an interval $[0,L_e]$.  Then, $e$ is endowed with a Schr\"odinger operator $-d^2/dx^2+q_e(x)$ ($0<x<L_e$).  A~global operator on $\G$ is determined by coupling these edge operators by a Robin condition at each vertex~$v$,
\begin{equation}\label{robin}
  \sum_{e\in\E(v)} f'(v) \;=\; \alpha_v f(v),
\end{equation}
which is called the Kirchoff or Neumann condition when $\alpha_v=0$.  The sum is over all edges incident to~$v$, and the prime denotes the derivative in the direction away from the vertex into the edge.  If all $\alpha_v$ are real, then (\ref{robin}) determines a self-adjoint operator $A$ in~$L^2(\Gamma)$.  The domain of $A$ consists of all functions in $L^2(\Gamma)$ whose restriction to each edge $e$ is in the Sobolev space $H^2(e)$ and that satisfy the Robin condition at each vertex.  The pair $(\Gamma,A)$ is a quantum graph, and it is periodic provided that the coordinatizations of the edges and the potentials~$q_e$ are invariant under the shift group~$\ZZ^d$.
All vertex conditions that correspond to self-adjoint operators in~$L^2(\Gamma)$ are described in~\cite[Theorem 1.4.4]{BerkolaikoKuchment2013}.

The next two steps are (1) the discrete-graph (or combinatorial) reduction of the equation $(A-\lambda)u=0$ and (2) the Floquet (discrete Fourier) transform; both are described in detail in~\cite[\S3.1-3.2]{Shipman2019}.  Solutions to $(A-\lambda)u=0$ are eigenfunctions of $A$ but cannot lie in $L^2(\Gamma)$; so $A$ is considered as a differential operator on the space of all functions that are in $H^2(e)$ of each edge and satisfy the Robin vertex conditions.  To say that $A$ is a periodic operator is to say that it commutes with the $\ZZ^d$ action.  A simultaneous eigenfunction $u$ of $A$ and the $\ZZ^d$ action is called a Floquet mode of $A$,
\begin{align}
  Au &= \lambda u \\
  u(gx) &= z^g u(x),
\end{align}
in which $z=(z_1,\dots,z_d)=(e^{ik_1},\dots,e^{ik_d})$ is the vector of Floquet multipliers (eigenvalues of the elementary shifts) and $z^g=\prod_{i=1}^d z_i^{g_i}$.
The first equation allows one to use the Dirichlet-to-Neumann map for the ODE $-u''+q_e(x)u-\lambda u=0$ on each edge to write the equation $(A-\lambda)u=0$ solely in terms of the values of $u$ on~$\V(\Gamma)$ by using~(\ref{robin}).  The second equation allows one to restrict the solution to a single fundamental domain.  The result is a homogeneous system of linear equations depending on $\lambda$ and $z$ for the values of $u$ on the finite set of vertices in a fundamental domain,
\begin{equation}
  \hat A(z,\lambda) u = 0.
\end{equation}
Here, $u$ is reused to denote the vector of values of the function $u$ on the vertices and $\hat A(z,\lambda)$ is a square matrix, which we call the {\em spectral matrix} of the quantum graph $(\G,A)$.
One should be a bit careful with using the word ``the", as the spectral matrix does depend on the choice of fundamental domain, but of course the Floquet modes of $(\G,A)$ corresponding to null vectors of $\hat A(z,\lambda)$ are independent of this choice.
We now describe explicitly how this matrix is constructed.

Let $\V_0$ and $\E_0$ denote the vertices and edges of a fixed fundamental domain of a periodic quantum graph $(\Gamma,A)$.  The matrix $\hat A(z,\lambda)$ is indexed by the vertices $\V_0$ minus those that have the Dirichlet condition. 
$\hat A(z,\lambda)$ is constructed as follows.  Given an edge $e\in \E_0$, there are vertices $v,w\in \V_0$ and $g\in\ZZ^d$ such that $e=\{v,gw\}$.
Let $e$ be parameterized by a variable $x\in[0,L]$ such that $x=0$ corresponds to $v$ and $x=L$ corresponds to~$gw$.  Let the transfer matrix for $-d^2u/dx^2 + q_e(x)u = \lambda u$ be
\begin{equation}
  T(\lambda)\;=\;\mat{1.2}{c(\lambda)}{s(\lambda)}{c'(\lambda)}{s'(\lambda)},
\end{equation}
that is, $T(\lambda)[u(0),du/dx(0)]^t = [u(L),du/dx(L)]^t$.
If both $v$ and $w$ have Robin or Neumann conditions and $v\!\not=\!w$, then the following matrix goes into the $2\!\times\!2$ submatrix of $\hat A(z,\lambda)$ indexed by $v$ and $w$:
\begin{equation}\label{DtNz}
  N(g,\lambda) \;=\; \frac{1}{s(\lambda)} \mat{1.2}{-c(\lambda)}{z^g}{z^{-g}}{-s'(\lambda)};
\end{equation}
and if $v\!=\!w$, then $(z^g+z^{-g}-c(\lambda)-s'(\lambda))$ goes into the diagonal entry indexed by $v$.
If $w$ has the Dirichlet condition, then $-c(\lambda)/s(\lambda)$ goes in the diagonal entry for $v$.  A diagonal matrix with entries $-\alpha_v$ is added to account for the Robin parameters for all the vertex conditions.  The matrix $N(0,\lambda)$ represents the Dirichlet-to-Neumann map for $e$ because it takes $(f(v),f(w))$ to $(f'(v),f'(w))$, where the prime denotes the derivative taken at a vertex directed into the edge.  If both $v$ and $w$ have the Dirichlet condition, then $\hat A(z,\lambda)$ attains an extra factor of $s(\lambda)$.

The matrix $\hat A(z,\lambda)$ is a Laurent polynomial in $z$ with coefficients that are matrix-valued meromorphic functions of $\lambda$.  The poles lie at the roots of the functions $s(\lambda)$ for all edges.  This set is denoted by
\begin{equation}
  \sigma_D(\G,A) \;=\; \left\{ \lambda : \exists\,e\in\E(\G),\,s_e(\lambda)=0 \right\}.
\end{equation}
Observe that the function $s_e(\lambda)$ is independent of the orientation of the parameterization of $e$ by $[0,L_e]$.  We call $s_e(\lambda)$ the {\em Dirichlet spectral function} for~$e$, as its roots are the Dirichlet eigenvalues of~$e$.

\smallskip
{\bfseries\slshape Remark.} The poles of $\hat A(z,\lambda)$ can be moved in the sense that, for each $\lambda$, there is an equivalent quantum graph $(\dot\G,\dot A)$ such that $\lambda\not\in\sigma_D(\dot\G,\dot A)$.  This graph is obtained by inserting artificial vertices in the interior of edges of $\G$, as described in~\cite[\S IV]{KuchmentZhao2019a}, and the spectral matrices for $(\dot\G,\dot A)$ and $(\G,A)$ are related by a meromorphic factor.  The following result is proved in section~\ref{sec:dotted}.

\begin{proposition}\label{prop:poles}
  Let $(\Gamma,A)$ be a periodic quantum graph, let $e=(v_1,v_2)$ be an edge of $\Gamma$, and let $\dot\Gamma$ be the graph obtained by placing an additional vertex $v$ in the interior of $e$, thus dividing $e$ into two edges $e_1=(v_1,v)$ and $e_2=(v,v_2)$, with the potentials on $e_1$ and $e_2$ being inherited from $q_e$ on~$e$.  Let $s(\lambda)$, $s_1(\lambda)$, and $s_2(\lambda)$ be the Dirichlet spectral functions for the edges $e$, $e_1$, and $e_2$.  Then
\begin{equation}
  s_1(\lambda)s_2(\lambda)\, D_{(\dot\Gamma,A)}(z,\lambda) \;=\;
  \pm s(\lambda)\, D_{(\Gamma,A)}(z,\lambda).
\end{equation}
\end{proposition}

%%%%%%%%%%%%%%%%%%%%%%%%%%%%%%%%%%%%%%%%%%%%%%%%%
\subsection{The Fermi surface}\label{subsec:Fermi}

The dispersion function for a periodic quantum graph $(\G,A)$ is defined as
\begin{equation}
  D_{(\Gamma,A)}(z,\lambda) \;:=\; \det \hat A(z,\lambda),
\end{equation}
and its zero set is the set of all $z$-$\lambda$ pairs at which $(\G,A)$ admits a Floquet mode.  By fixing an energy $\lambda\in\CC$, one obtains the {\em Floquet surface} of $(\G,A)$:
\begin{equation}
  \Phi_\lambda = \Phi_{(\Gamma,A),\lambda} =
  \big\{ z\in(\CC^*)^d : D_{(\Gamma,A)}(z,\lambda) = 0 \big\}.
\end{equation}
When considered as a set of wavevectors $(k_1,\dots,k_d)\in\CC$ (with $z_j=e^{ik_j}$), it is the {\em Fermi surface} of $(\G,A)$.  We will just use the term ``Fermi surface" for $\Phi_\lambda$.
The spectrum of $(\Gamma,A)$ consists of all energies~$\lambda$ such that the Fermi surface intersects the $d$-torus $\TT^d = \{ z\in\CC^d : |z_1|=\dots=|z_d|=1 \}$,
\begin{equation}
  \sigma_{(\Gamma,A)} = \left\{ \lambda\in\CC : \Phi_{(\Gamma,A),\lambda} \cap \TT^d \not= \emptyset \right\}.
\end{equation}
Both $D(z,\lambda)$ and $\Phi_\lambda$ are independent of the choice of fundamental domain of~$\G$.

Importantly, when $\Gamma$ is disconnected, with each connected component being a compact graph, a fundamental domain can be chosen to be one component $\G_0$, and thus $\hat A(z,\lambda)$ and $D(z,\lambda)$ are independent of~$z$.  All of the matrices~(\ref{DtNz}) have $g\!=\!0$ and reduce to the Dirichlet-to-Neumann maps for the edges.  In this case, the spectral matrix, which can be denoted by $\hat A(\lambda)$ is the spectral matrix of $A$ confined to the finite graph $\G_0$, and the roots of its determinant $D(\lambda)$ are the eigenvalues of this finite quantum graph.

The Fermi surface is an algebraic set in $(\CC^*)^d$, and it is {\em\bfseries reducible} at $\lambda$ whenever $\Phi_\lambda$ is the union of two algebraic sets.  This occurs whenever $D(z,\lambda)$ is factorable into two polynomials, neither of which is a monomial.  One should be aware of the situation when $D(z,\lambda)=D_1(z,\lambda)^m$, with $D_1(z,\lambda)$ being irreducible; $\Phi_\lambda$ is reducible on account of having a component of multiplicity~$m$.

%%%%%%%%%%%%%%%%%%%%%%%%%%%%%%%%%%%%%%%%%%%%%%%%%
\subsection{A calculus for joining two periodic graphs}

The lemma in this section is the building block for the subsequent theorems on reducible Fermi surfaces for multi-layer quantum graph operators.
It can be viewed as a sort of periodic-graph version of the surgery principles for finite quantum graphs~\cite{BerkolaikoKennedyKurasov2019a}.  These describe how the spectrum of a new graph is related to the spectra of old graphs under various modifications and joinings.
It will be convenient to use the following notation for the dispersion function of a periodic quantum graph~$(\Gamma,A)$, which emphasizes the dependence on $\Gamma$,
\begin{equation}
  \left[ \Gamma \right] := D_{(\Gamma,A)}(z,\lambda).
\end{equation}

\begin{definition}[$\Gamma^v$]\label{def:Gammav}
Let $\Gamma$ be a $d$-periodic graph, and let $v$ be a vertex of $\Gamma$ of degree $r$.
Denote by $\Gamma^v$ the periodic graph obtained by replacing $gv$ by $r$ terminal vertices incident to the $r$ edges that are incident to $v$ in~$\Gamma$.
\end{definition}

\noindent
If a Schr\"odinger operator $A$ is defined on $\Gamma$ as a metric graph, define a Schr\"odinger operator $A^v$ on $\G^v$ as follows.

\begin{definition}[($\Gamma^v,A^v)$]\label{def:GammaAv}
  Let $(\Gamma,A)$ be a $d$-periodic quantum graph containing vertex $v\in\V(\Gamma)$.  Denote by $(\Gamma^v,A^v)$ the quantum graph obtained by replacing the vertex condition at each vertex in the orbit $\{gv : g\in\ZZ^d\}$ in $\Gamma$ with the Dirichlet condition.  Thus $\G$ may as well be replaced by~$\G^v$.
\end{definition}

\begin{definition}[Single-vertex join $\Gamma_1(v_1\,v_2)\Gamma_2$]\label{def:svj}
  Let $\Gamma_1$ and $\Gamma_2$ be $d$-periodic quantum graphs with Robin parameter $\alpha_1$ at $v_1\in\V(\Gamma_1)$ and $\alpha_2$ at $v_2\in\V(\Gamma_2)$.  The {\em single-vertex join} of $\Gamma_1$ and $\Gamma_2$ at the pair $(v_1,v_2)$, denoted by $\Gamma_1(v_1\,v_2)\Gamma_2$, is a quantum graph with vertex set $\V(\Gamma_1)\cup\V(\Gamma_2)/\!\equiv$, in which $gv_1\equiv gv_2$ for all $g\in\ZZ^d$ and edge set $\E(\Gamma_1)\cup\E(\Gamma_2)$.  A Robin vertex condition with parameter $\alpha_1+\alpha_2$ is imposed at the joined vertices $gv_1\equiv gv_2$, and all other vertex conditions are inherited from $\Gamma_1$ and~$\Gamma_2$.
If the Robin parameter at the joined vertex $v_1\equiv v_2$ is changed to $\alpha$, the resulting graph is denoted by
$\Gamma_1(v_1\,v_2)_{\!\alpha\,}\Gamma_2$.
\end{definition}

\begin{lemma}\label{lemma:svj}
  Let $\Gamma_1$ and $\Gamma_2$ be $d$-periodic quantum graphs with $v_1\in\V(\Gamma_1)$ and $v_2\in\V(\Gamma_2)$.  Then the dispersion function for $\Gamma_1(v_1,v_2)\Gamma_2$ is
\begin{equation}
  \left[ \Gamma_1(v_1,v_2)\Gamma_2 \right]
    \;=\; \left[ \Gamma_1 \right]\left[ \Gamma_2^{v_2} \right]
    +  \left[ \Gamma_1^{v_1} \right] \left[ \Gamma_2 \right],
\end{equation}
and the dispersion function for $\Gamma_1(v_1\,v_2)_{\!\alpha\,}\Gamma_2$ is
\begin{equation}
  \left[ \Gamma_1(v_1\,v_2)_{\!\alpha\,}\Gamma_2 \right]
    = \left[ \Gamma_1 \right]\left[ \Gamma_2^{v_2} \right]
    +  \left[ \Gamma_1^{v_1} \right] \left[ \Gamma_2 \right]
    + (\alpha-\alpha_1-\alpha_2) \left[ \Gamma_1^{v_1} \right] \left[ \Gamma_2^{v_2} \right].
\end{equation}
\end{lemma}

\begin{proof}
Let $A_1$ and $A_2$ be the operators associated with the quantum graphs $\Gamma_1$ and $\Gamma_2$, and let $\hat A_1(z,\lambda)$ and $\hat A_2(z,\lambda)$ be the spectral matrices of these operators.  Let $\hat A_1^0(z,\lambda)$ and $\hat A_2^0(z,E)$ be the spectral matrices of the operators associated with $\Gamma_1^{v_1}$ and $\Gamma_2^{v_2}$.  By ordering the vertices of a fundamental domain of $\Gamma_1$ such that $v_1$ is listed last, and ordering the vertices of a fundamental domain of $\Gamma_2$ such that $v_2$ is listed first, one obtains the block decomposition
\begin{equation}
  \hat A_1 = 
\renewcommand{\arraystretch}{1.1}
\left[
\begin{array}{cc}
  \hat A_1^0 & a_1 \\ a_1^* & a_1^0
\end{array}
\right], \qquad
  \hat A_2 = 
\renewcommand{\arraystretch}{1.1}
\left[
\begin{array}{cc}
  a_2^0 & a_2^* \\ a_2 & \hat A_2^0
\end{array}
\right],
\end{equation}
in which $a_1$ and $a_2$ are column vectors and $a_1^0$ and $a_2^0$ are scalars.

The matrix $\hat A$ of the operator associated with $\Gamma_1(v_1\,v_2)\Gamma_2$ is
\begin{equation}\label{hatA}
  \renewcommand{\arraystretch}{1.1}
  \hat A \;=\; 
\left[
\begin{array}{ccc}
  \hat A_1^0 & a_1 & \mathbf{0} \\
  a_1^* & a_1^0 + a_2^0 & a_2^* \\
  \mathbf{0} & a_2 & \hat A_2^0
\end{array}
\right].
\end{equation}
Notice that the entry $a_1^0 + a_2^0$ incorporates the Robin parameter $\alpha_1+\alpha_2$.
The first statement of the theorem is just the following statement about determinants:
\begin{equation*}
  \det(\hat A) \;=\; \det(\hat A_1)\det(\hat A_2^0) + \det(\hat A_1^0)\det(\hat A_2).
\end{equation*}
The matrix Floquet transform for $\Gamma_1(v_1\,v_2)_{\!\alpha\,}\Gamma_2$ is obtained by adding $\alpha-\alpha_1-\alpha_2$ to the term $a_1^0+a_2^0$ in~(\ref{hatA}), and the second statement of the theorem follows.
\end{proof}

%%%%%%%%%%%%%%%%%%%%%%%%%%%%%%%%%%%%%%%%%%%%%%%%%
\subsection{Separable periodic graphs}

The class of multi-layer graphs we call type 1 is built on separable layers.
Each layer has the property that, when severed periodically at a certain vertex, it falls apart into a $d$-dimensional array of identical finite graphs, as illustrated in Fig.\,\ref{fig:Separable}.  This is made precise as follows.

\begin{definition}[separable periodic graph]\label{def:separable}
A $d$-periodic graph $\G$ is {\em separable at $v\in\V(\G)$} if $\Gamma^v$ is the union of the $\ZZ^d$ translates of a finite  graph, or, equivalently, if $\G^v$ has compact connected components.
\end{definition}

\begin{figure}[h]
\centerline{
\scalebox{0.25}{\includegraphics{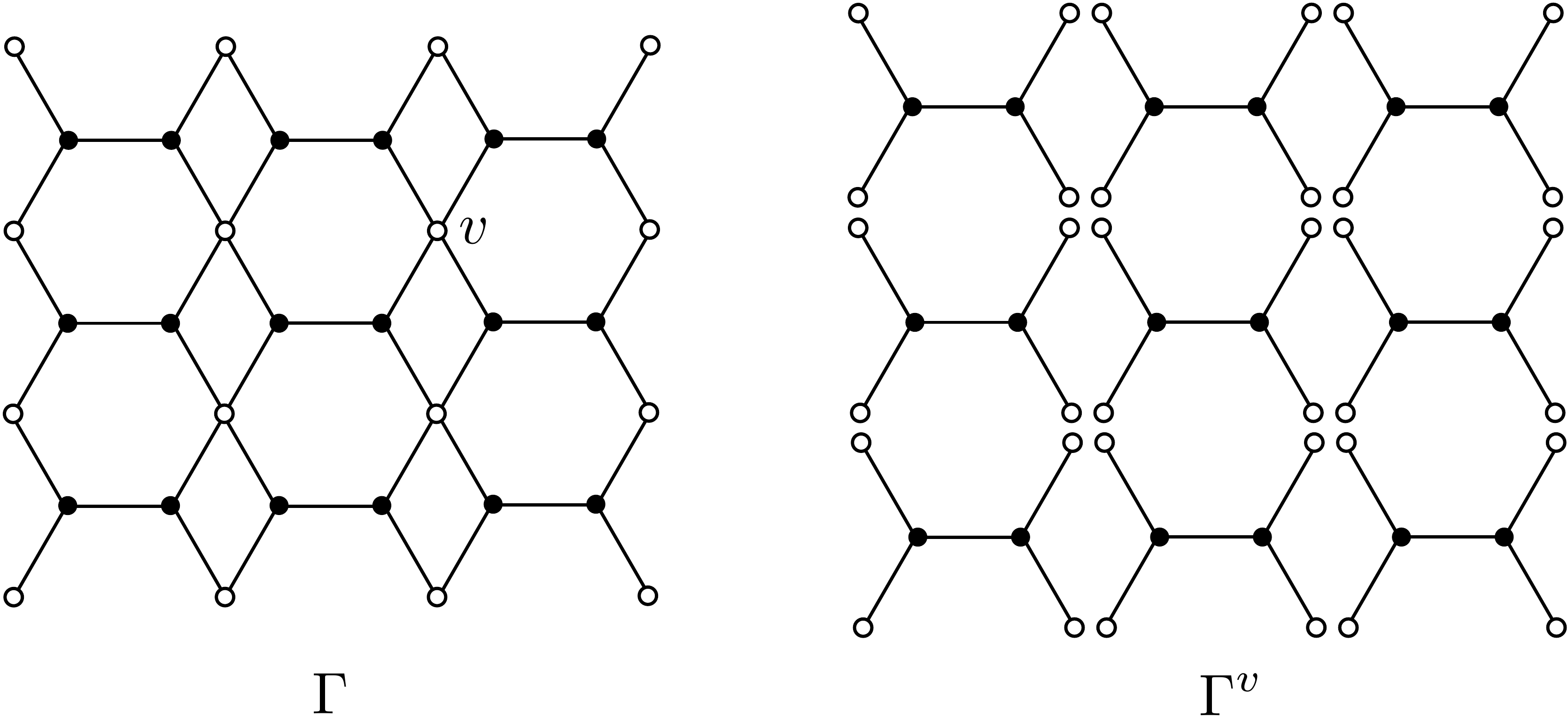}}
}
\caption{\small A $2$-periodic graph that is separable at a vertex $v$; and the corresponding separated graph $\Gamma^v$.}
\label{fig:Separable}
\end{figure}

%%%%%%%%%%%%%%%%%%%%%%%%%%%%%%%%%%%%%%%%%%%%%%%%%
%%%%%%%%%%%%%%%%%%%%%%%%%%%%%%%%%%%%%%%%%%%%%%%%%
\section{Type 1: Multi-layer graphs with separable layers}\label{sec:type1}

This section develops a class of multi-layer graphs whose individual layers are separable and whose Fermi surface is reducible, with several components.  A~1-periodic illustration is in Fig.\,\ref{fig:Layered}(left).  An example is AB-stacked graphene, which is discussed in Section~\ref{subsec:AB}.

%%%%%%%%%%%%%%%%%%%%%%%%%%%%%%%%%%%%%%%%%%%%%%%%%
\subsection{Type-1 multi-layer graphs}

First, we define the {\em layers} (black in Fig.~\ref{fig:Layered}(left)) and the {\em connector graph} (blue in Fig.~\ref{fig:Layered}(left)).
Let $\zeta(z,\lambda)$ be a Laurent polynomial in $z=(z_1,\dots,z_d)$ with coefficients that are meromorphic in $\lambda$.
The $j$-th layer ($j=1,\dots,n$) is a $d$-periodic quantum graph $(\Lambda_j,A_j)$, with a distinguished vertex $v_j$ and such that the dispersion function $d_j(z,\lambda)$ of $(\Lambda_j,A_j)$ and the dispersion function $\mathring{d}_j(z,\lambda)$ of $(\Lambda_j^{v_j},A_j^{v_j})$ are polynomial functions of $\zeta(z,\lambda)$ with coefficients that are meromorphic in~$\lambda$,
\begin{align}
  d_j(z,\lambda) &\;=\; p_j\big(\zeta(z,\lambda),\lambda\big), \\
  \mathring{d}_j(z,\lambda) &\;=\; \mathring{p}_j\big(\zeta(z,\lambda),\lambda\big).
\end{align}
Of course, we have in mind specifically the situation in which each layer $\Lambda_j$ is separable at~$v_j$.  This is because, in this case, $\Lambda_j^{v_j}$ is a disjoint union of compact graphs, and thus its dispersion function is a meromorphic function $f_j(\lambda)$ and therefore a degree-0 polyomial in $\zeta(z,\lambda)$,
\begin{equation}
   \mathring{d}_j(z,\lambda) \;=\; f_j(\lambda).
\end{equation}
It will be useful to allow $\mathring{p}_j$ to be a more general polynomial in applications such as ABC-stacked graphene in Section~\ref{sec:graphene}.  The connector graph is a finite quantum graph $(\Sigma,B)$, together with a list of distinct vertices $w_j\in\V(\Sigma)$ ($j=1,\dots,n$).

From these ingredients, an $n$-layer quantum graph $(\Gamma,A)$ of type 1, as illustrated on the left of Fig.\,\ref{fig:Layered}, is formed as follows.  The vertex $v_j$ is merged, or identified, with $w_j$.  
In like manner, for each $g\in\ZZ^d$, the translated vertices $gv_1,\dots,gv_n$ are coupled by another copy of $\Sigma$, called $g\Sigma$.  The resulting periodic graph~$\G$ is what we call a type-1 multi-layer graph.  Its edge set consists of the edges of each layer $\Lambda_j$ and the edges of each translate $g\Sigma$ of~$\Sigma$.  By denoting the identification of merged vertices by the equivalence relation~$\equiv$, the vertex and edge sets of $\Gamma$~are
\begin{align}
  \V(\Gamma) &\;=\; \left( \bigcup\limits_{j=1}^n \!\V(\Lambda_j) \;\cup\; \bigcup\limits_{g\in\ZZ^d} \!\V(g\Sigma) \right)/\equiv \\
  \E(\Gamma) &\;=\; \bigcup\limits_{j=1}^n \!\E(\Lambda_j) \;\cup\; \bigcup\limits_{g\in\ZZ^d} \!\E(g\Sigma).
\end{align}

The Schr\"odinger operator $A$ on $\Gamma$ has the same differential-operator expression as the operators $A_j$ and~$B$ on the elemental graphs $\Lambda_j$ and~$\Sigma$, and the Robin parameter of an equivalence class of merged vertices (now a single vertex of $\Gamma$) is assigned the sum of the Robin parameters of all the vertices that were merged.

Observe that it is possible to allow several of the vertices $w_j$ to be equal.  In this case, one might as well merge all the layers that are attached to that vertex into a single layer according to the single-vertex join in Definition~\ref{def:svj}, applied several times.  According to the calculus of Lemma~\ref{lemma:svj}, the degree of the polynomial $p_j$ for this new layer is the maximum of the degrees of the polynomials of the joined components. 

\begin{theorem}\label{thm:type1}
  Let $(\Gamma,A)$ be an $n$-layer $d$-periodic quantum graph of type~1.  Its dispersion function $D(z,\lambda)$ is a polynomial in $\zeta(z,\lambda)$ with coefficients that are meromorphic functions of $\lambda$,
\begin{equation}
  D(z,\lambda) \;=\; P(\zeta(z,\lambda),\lambda),
\end{equation}
and the degree of $P$ as a polynomial in $\zeta$ is
\begin{equation}
  \deg P \;=\; \sum_{j=1}^n \deg p_j.
\end{equation}
\end{theorem}

The proof applies Lemma~\ref{lemma:svj} iteratively as successive layers are connected through the connector graph.

\begin{proof}
When the number of layers is zero, $(\G,A)$ is the union $\cup_{g\in\ZZ^d} g\Sigma$ of disconnected finite components with the operator $B$ acting on each component.  The dispersion function is a meromorphic function of~$\lambda$, independent of $z$, and is thus trivially a polynomial in $\zeta(z,\lambda)$ of degree $0$, with coefficients that are meromorphic in~$\lambda$.
As the induction hypothesis, let the theorem hold with $n$ replaced by $n\!-\!1$, with $n\geq1$.

Let $(\G,A)$ be the type-1 $n$-layer quantum graph supposed in the theorem, with layers
$(\Lambda_j,A_j)$ separable at $v_j$ ($1\leq j\leq n$) and connector graph $(\Sigma,B)$ with distinct joining vertices $\{w_j\}_{j=1}^n$.  If any of the polynomials $p_j$ has degree $0$, then $\Lambda_j$ is a disjoint union of $\ZZ^d$ translates of a finite graph, and this finite graph might as well be joined with the connector graph $\Sigma$.  Therefore, we assume that each $\deg p_j\geq 1$ for all $j:1\leq j\leq n$.

Denote by $(\tilde\G,\tilde A)$ the type-1 $(n\!-\!1)$-layer quantum graph built from the layers $\{(\Lambda_j,A_j)\}_{j=1}^{n-1}$ and the connector objects $(\Sigma,B)$ and $\{w_j\}_{j=1}^{n-1}$.  Note that $(\tilde\G^{w_n},\tilde A^{w_n})$ is the type-1 $(n\!-\!1)$-layer quantum graph built from the layers $\{(\Lambda_j,A_j)\}_{j=1}^{n-1}$ and the connector objects $(\Sigma^{w_n},B^{w_n})$ and $\{w_j\}_{j=1}^{n-1}$.  Denote by $\tilde D(z,\lambda)$ and $\tilde D^0(z,\lambda)$ the dispersion functions of $(\tilde\G,\tilde A)$ and $(\tilde\G^{w_n},\tilde A^{w_n})$.  By the induction hypothesis, they are polynomials in $\zeta(z,\lambda)$ with coefficients that are meromorphic in~$\lambda$, and both are of degree~$\sum_{j=1}^{n-1}\deg p_j$.

The graph $(\G,A)$ is the single-vertex join of $(\tilde\G,\tilde A)$ and $\Lambda_n$,
\begin{equation}
  \G \;=\; \tilde\G(w_n,v_n)\Lambda_n
\end{equation}
and the calculus of Lemma~\ref{lemma:svj} yields
\begin{equation}\label{hi}
  [\G] \;=\; [\tilde\G] [\Lambda_n^{v_n}] + [\tilde\G^{w_n}] [\Lambda_n].
\end{equation}
Since $\Lambda_n$ is separable at $v_n$, $[\Lambda_n^{v_n}]$ is independent of $z$, so the degree of the first term on the right-hand side of (\ref{hi}), as a polynomial in $\zeta$, is $m=\sum_{j=1}^{n-1}\deg p_j$.  The degree of the second term as a polynomial in $\zeta$ is $m+\deg p_n$.  This completes the induction.
\end{proof}

A special case of this theorem occurs when there is only one layer.  The connector graph $\Sigma$ is then viewed as a periodic ``decoration" of $\Lambda_1$.  The result is the following corollary.  Much more is known about decorated periodic graphs, particularly with regard to opening spectral gaps~\cite{SchenkerAizenman2000}.

\begin{corollary}[Decorated graphs]
  Let $(\Gamma,A)$ be a $d$-periodic quantum graph that is separable at vertex~$v$, and let
  $(\Sigma,B)$ be a finite decorator graph with distinguished vertex $w\in\V(\Sigma)$.  Let $\ell(\lambda)$ denote the spectral function of $(\Gamma^v,A^v)$, and let $h(\lambda)$ and $h^0(\lambda)$ denote the spectral functions of $(\Sigma,B)$ and $(\Sigma^w,B^w)$.
  
  Denote by $(\bar\G,\bar A)$ the ``decorated graph" obtained by the single-vertex join of $(\G,A)$ and
 $\Delta=\cup_{g\in\ZZ^d} g\Sigma$ at the vertices $v$ and $w$.
 If the Fermi surface of $(\G,A)$ at energy $\lambda$ is given by
\begin{equation}
  D(z,\lambda) \;=\; 0,
\end{equation}
then the Fermi surface of $(\bar\G,\bar A)$ at $\lambda$ is given by
\begin{equation}
  D(z,\lambda) \;=\; -\frac{\ell(\lambda)h(\lambda)}{h^0(\lambda)}.
\end{equation}
\end{corollary}

\begin{proof}
  Actually, this is a bit more than just a corollary to the theorem.  The theorem says that $[\bar\G]$ is a function of $D(\lambda,z)$ that is linear in $D(z,\lambda)$ (take $\zeta(z,\lambda)=D(z,\lambda)$).  But we can find the coefficients from Lemma~\ref{lemma:svj}, which says
  $[\bar\G] \;=\; [\G] [\Delta^w] + [\G^v] [\Delta]$, or
\begin{equation}
  [\bar\G] \;=\; D(z,\lambda) h^0(\lambda) + \ell(\lambda) h(\lambda).
\end{equation}
The Fermi surface of $(\bar\G,\bar A)$ is $[\bar\G]=0$, from which follows the result.
\end{proof}

\bigskip

The polynomial $P(\zeta,\lambda)$ in Theorem~\ref{thm:type1} factors into $m = \sum_{j=1}^n \deg p_j$ linear factors as a function of $\zeta$, and each factor corresponds to a component of the Fermi surface of~$(\G,A)$.

\begin{theorem}\label{thm:reducible1}
  The Fermi surface of a type-1 $n$-layer $d$-periodic quantum graph is reducible into
 $m = \sum_{j=1}^n \deg p_j$ components (with possible multiplicities).  Each component is of the form
\begin{equation}\label{type1fermi}
  \zeta(z,\lambda) \;=\; \mu(\lambda).
\end{equation}
For each $\lambda$, the $m$ (not necessarily distinct) values of $\mu(\lambda)$ are the roots of~$P(\zeta,\lambda)$.
\end{theorem}

%%%%%%%%%%%%%%%%%%%%%%%%%%%%%%%%%%%%%%%%%%%%%%%%%
\subsection{Coupling several type-1 multi-layer graphs}\label{subsec:several}

Multi-layer graphs themselves can be used as the layers of more complex multi-layer graphs.  This ideas will be used for ABC-stacked graphene in section~\ref{subsec:ABC}.  Observe that single-layer graphene is separable at each of the two vertices of a fundamental domain.

Let $(\Gamma_j,A_j)$ be $d$-periodic type-1 multi-layer graphs based on separable layers described in the previous section, with a common composite Floquet variable~$\zeta(z,\lambda)$.  Suppose that each $(\Gamma_j,A_j)$ has a distinguished layer $(\Lambda_j,B_j)$ such that $\Lambda_j$ is separable at a vertex $v_j$ other than the one used in constructing~$(\Gamma_j,A_j)$.  The dispersion function of $(\Lambda_j^{v_j},B_j^{v_j})$ is therefore independent of $z$ and can be used as a layer in a type-1 multi-layer graph.  By replacing the layer $(\Lambda_j,B_j)$ in the construction of $(\Gamma_j,A_j)$ by the layer $(\Lambda_j^{v_j},B_j^{v_j})$, one obtains $(\Gamma_j^{v_j},A_j^{v_j})$.  The point here is that, by Theorem~\ref{thm:type1}, both $(\Gamma_j,A_j)$ and $(\Gamma_j^{v_j},A_j^{v_j})$ have dispersion function that is a polynomial in~$\zeta(z,\lambda)$.  The theorem is then applied again with $(\Gamma_j,A_j)$ as the layers, which are coupled by an arbitrary finite connector graph~$(\Sigma,C)$ by joining the vertices $v_j$ with given vertices $w_j$ of~$\Sigma$.

%%%%%%%%%%%%%%%%%%%%%%%%%%%%%%%%%%%%%%%%%%%%%%%%%
%%%%%%%%%%%%%%%%%%%%%%%%%%%%%%%%%%%%%%%%%%%%%%%%%
\section{Type 2: Multi-layer graphs with bipartite layers}\label{sec:type2}

This section generalizes the construction in~\cite[\S6]{Shipman2019} from bi-layer to $n$-layer quantum graphs and from single-edge coupling to coupling by general graphs, as illustrated on the right in Fig.\,\ref{fig:Layered}.  Each layer has the same underlying graph, which is bipartite with exactly one ``red" and one ``green" vertex in a fundamental domain.
The layers are connected by one graph connecting the $n$ red vertices in a fundamental domain and another graph connecting the $n$ green vertices in a fundamental domain.  A topical example is AA-stacked graphene, which is discussed in Section~\ref{subsec:AA}.

%%%%%%%%%%%%%%%%%%%%%%%%%%%%%%%%%%%%%%%%%%%%%%%%%
\subsection{Coupling by arbitrary finite graphs}

Given that the quantum graph $(\Lambda,\Ao)$, for a given layer, has underlying graph $\Lambda$ which is bipartite with one red and one green vertex per period, the Floquet transform of its discrete reduction is a $2\!\times\!2$ matrix
\begin{equation}
  \hat {\Ao}(z,\lambda) = \mat{1.2}{b_1(\lambda)}{w(z,\lambda)}{w(z^{-1},\lambda)}{b_2(\lambda)},
\end{equation}
in which $b_i(\lambda)$ are meromorphic functions of $\lambda$ and $w(z,\lambda)$ is a Laurent polynomial in $z=(z_1,\dots,z_d)$ with coefficients that are meromorphic in $\lambda$.  (See section~\ref{sec:graphenesingle} for the case of graphene.) Specifically, $w(z,\lambda)$ is a sum over some finite subset $Z\subset\ZZ^d$,
\begin{equation}\label{wzlambda}
  w(z,\lambda) \;=\; \sum_{\ell\in Z} \frac{z^\ell}{s_\ell(\lambda)}\,,
\end{equation}
in which $s_\ell(\lambda)$ is the $s$-function for the potential $q(x)$ on the edge connecting a green vertex in a given fundamental domain with a red vertex in the domain shifted by $\ell\in\ZZ^d$, and $z^\ell = z_1^{\ell_1}\cdots z_n^{\ell_n}$.

The requirement for the multi-layer graphs in the theorem below is that $w(z,\lambda)$ be identical over all the layers; but the functions $b_1(\lambda)$ and $b_2(\lambda)$ are allowed to vary from layer to layer.  This means that each layer must have the same underlying graph $\Lambda$, and that for any given edge of $\Lambda$, the potential at each layer must have the same $s(\lambda)$-function, or, equivalently, the potentials must have the same Dirichlet spectrum.  Indeed, the Dirichlet spectrum of $-d^2/dx^2 + q(x)$ on an interval and the function $s(\lambda)$ determine each other; 
see~\cite[Ch.\,2 Theorem~5]{PoschelTrubowitz1987}, for example.

Several such graphs $(\Lambda,A_k)$, $k=1,\dots,n$, with the same $w(z,\lambda)$, are coupled to form an $n$-layer quantum graph $(\Gamma,A)$ as depicted on the right in Fig.\,\ref{fig:Layered}.  Start with the disjoint union of the $n$ graphs $(\Lambda,A_k)$.  Then replace the $n$ red vertices in a fundamental domain with a finite quantum graph $(\Sigma_1,B_1)$ whose vertex set includes those $n$ red vertices plus (possibly) additional ones.  Another finite graph $(\Sigma_2,B_2)$ connects the green vertices together.  These two coupling graphs are repeated periodically.  The simplest case is when two successive single layers $(\Lambda,A_k)$ and $(\Lambda,A_{k+1})$ are connected with a single edge between corresponding red vertices and a single edge between corresponding green vertices.  The spectral matrix for $(\Gamma,A)$~is
\begin{equation}\label{Ahat2}
\begin{split}
   \hat A(z,\lambda)
   &\;=\; \mat{1.3}{\mathbf{b}_1(\lambda)}{w(z,\lambda) Q}{w(z^{-1},\lambda) Q^T}{\mathbf{b}_2(\lambda)}
      + \mat{1.3}{B_1(\lambda)}{0}{0}{B_2(\lambda)} \\
   &\;=\; \mat{1.3}{\tilde B_1(\lambda)}{w(z,\lambda) Q}{w(z^{-1},\lambda) Q^T}{\tilde B_2(\lambda)},
\end{split}
\end{equation}
in which $Q$ is the $m_1\times m_2$ matrix with the $n\times n$ identity matrix in its upper left, all other entries being zero; $\mathbf{b}_1(\lambda)$ ({\itshape resp.} $\mathbf{b}_2(\lambda)$) is a square diagonal matrix of size $n\!+\!m_1$ ({\itshape resp.} $n\!+\!m_2$),
\begin{equation}
  \mathbf{b}_1(\lambda) \;=\; \underset{j=1\dots n}{\mathrm{diag}} b_1^j(\lambda) \oplus 0_{m_1}
  \;=\; \renewcommand{\arraystretch}{1}
\left[
\begin{array}{cccccc}
  b_1^1(\lambda) & & & & & \\
  \vspace{-4ex} \\
  & \hspace{-8pt}\ddots\hspace{-8pt} & & & & \\
  & & b_1^n(\lambda) & & & \\
  & & & \hspace{-7pt}0 & & \\
  \vspace{-4ex} \\
  & & & & \hspace{-7pt}\ddots\hspace{-5pt} & \\
  & & & & & 0
\end{array}
\right],
\end{equation}
$B_1(\lambda)$ is the $m_1\times m_1$ spectral matrix of the coupling graph for the red vertices and the $m_2\times m_2$ matrix $B_2(\lambda)$ is for the green vertices.  $B_1(\lambda)$, for example, has the 

The dispersion function of $(\Gamma,A)$ is
\begin{align}
  D(z,\lambda) \,=\,  \det \hat A(z,\lambda) &\;=\;
      \det\big(\tilde B_1(\lambda)\big) \det\left( \tilde B_2(\lambda) - w(z,\lambda)w(z^{-1},\lambda)\, Q^T \tilde B_1(\lambda)^{-1} Q \right) \label{detAhat}\\
  &\;=\; P\big(w(z,\lambda)w(z^{-1},\lambda),\lambda\big), \label{Pzetalambda}
\end{align}
in which $P(\cdot,\lambda)$ is a polynomial of degree $n$ with coefficients that are meromorphic functions of~$\lambda$.
For a single layer, this polynomial is just a linear function of the composite Floquet variable
\begin{equation}\label{zetaww}
  \zeta(z,\lambda) := w(z,\lambda)w(z^{-1},\lambda).
\end{equation}
We have proved the following theorem.

\begin{theorem}[bipartite layers]\label{thm:type2}
  Let $(\Gamma,A)$ be a multi-layer periodic quantum graph obtained by coupling $n$~quantum graphs $(\Lambda,A_k)$, $k=1,\dots,n$, in which the underlying graph $\Lambda$ is bipartite with exactly one vertex of each ``color" in a fundamental domain and the potentials defining the operators $A_k$ have the same Dirichlet spectrum on corresponding edges and the coupling is effectuated by finite graphs of each color, as described for a type-2 multi-layer graph.  The Robin parameters may be different across layers.
  
For each energy $\lambda$, the Fermi surface of $(\Gamma,A)$ has $n$ components (which may have multiplicity greater than~$1$).  The components are of the form
\begin{equation}\label{zetaAA}
  \zeta(z,\lambda) \;=\; \rho(\lambda),
\end{equation}
in which $\rho(\lambda)$ is a root of the polynomial~(\ref{detAhat}).  The multiplicity of a component is equal to the multiplicity of the corresponding root.
\end{theorem}

%%%%%%%%%%%%%%%%%%%%%%%%%%%%%%%%%%%%%%%%%%%%%%%%%
\subsection{Generalization to decorated edges}

Type-2 multi-layer graphs can be generalized by replacing the edges of the bipartite graph with finite graphs that have two distinguished terminal edges, which we call ``decorated edges".  This is illustrated in Fig.\,\ref{fig:GrapheneDecorated} for the graphene structure.  A decorated edge admits a Dirichlet-to-Neumann map that straightforwardly generalizes that of an edge.  The DtN map uses the generalized $c(x,\lambda)$ and $s(x,\lambda)$ functions, whose value and derivative are prescribed at one terminal end and then computed at the other terminal end to form the transfer and DtN matrices.  When forming the spectral matrix $\hat A(z,\lambda)$, this DtN map is used, as described around equation~(\ref{DtNz}).

Incidentally, edges in general, for any of the quantum graphs we consider, may as well be decorated edges.  This can be thought of loosely as allowing a broader class of potentials on the edges.

\begin{figure}[h]
\centerline{
\scalebox{0.45}{\includegraphics{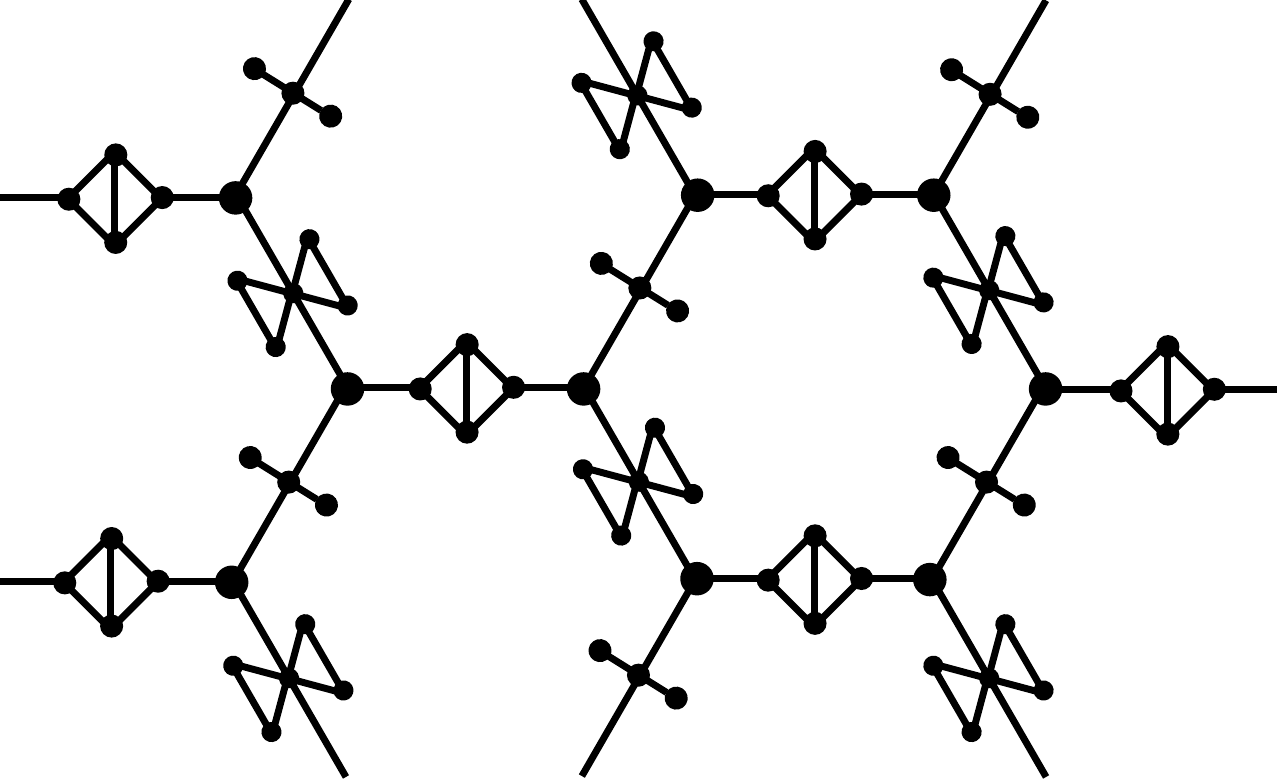}}
}
\caption{\small Quantum-graph graphene model with decorated edges.}
\label{fig:GrapheneDecorated}
\end{figure}

%%%%%%%%%%%%%%%%%%%%%%%%%%%%%%%%%%%%%%%%%%%%%%%%%
\subsection{Stacking by edges}\label{sec:iterative2}

When $n$ successive layers are connected by edges (or decorated edges), the matrix $Q$ in expression (\ref{detAhat}) becomes the $n\!\times\!n$ identity matrix $I_n$, and the dispersion function simplifies to
\begin{equation}\label{D2}
  D(z,\lambda) \;=\; \det\big( \tilde B_1(\lambda)\tilde B_2(\lambda) - \zeta(z,\lambda) I_n \big),
\end{equation}
and therefore the components of the Fermi surface are
\begin{equation}
  \zeta(z,\lambda) \;=\; \rho_j(\lambda),
  \quad j=1,\dots,n,
\end{equation}
where $\rho_j$ are the eigenvalues of the matrix~$\tilde B_1(\lambda)\tilde B_2(\lambda)$.

As the connector graphs are linear graphs, their spectral matrices $B_i(\lambda)$ are tridiagonal, with DtN matrices for the connector edges along the principal $2\!\times\!2$ submatrices.

%%%%%%%%%%%%%%%%%%%%%%%%%%%%%%%%%%%%%%%%%%%%%%%%%
%%%%%%%%%%%%%%%%%%%%%%%%%%%%%%%%%%%%%%%%%%%%%%%%%
\section{Multi-layer graphene}\label{sec:graphene}

We apply the theory developed in this work to quantum-graph models of multi-layer graphene structures.  Very general stacking of graphene, where the layers are shifted or rotated, results in a reducible Fermi surface.  We also discuss the conical singularities at wavevectors $(k_1,k_2)=\pm(2\pi/3,-2\pi/3)$ for single-layer graphene and how stacking multiple layers destroys them.

%%%%%%%%%%%%%%%%%%%%%%%%%%%%%%%%%%%%%%%%%%%%%%%%%
\subsection{The single layer}\label{sec:graphenesingle}

A graph model of graphene is hexagonal and bipartite, having two vertices and three edges of length~$1$ per fundamental domain.  Being bipartite, it is also separable at any vertex.

The most general quantum-graph model $(\Lo,\Ao)$ for which the differential operator on the edges is of the form $-d^2/dx^2 + q(x)$ features three potentials, one for each edge in a period, and two Robin parameters $\alpha_i$, one for each vertex $v_i$ ($i=1,2$) in a period.  The potentials will be denoted by $q_i(x)$ ($i=0,1,2$) as in Fig.\,\ref{fig:GrapheneS} and the corresponding transfer matrices by
\begin{equation}\label{Ti}
  T_i(\lambda) \;=\;
  \mat{1.3}{c_i(\lambda)}{s_i(\lambda)}{c_i'(\lambda)}{s_i'(\lambda)}
  \qquad
  (i=0,1,2).
\end{equation}

\begin{figure}[h]
\centerline{
\scalebox{0.3}{\includegraphics{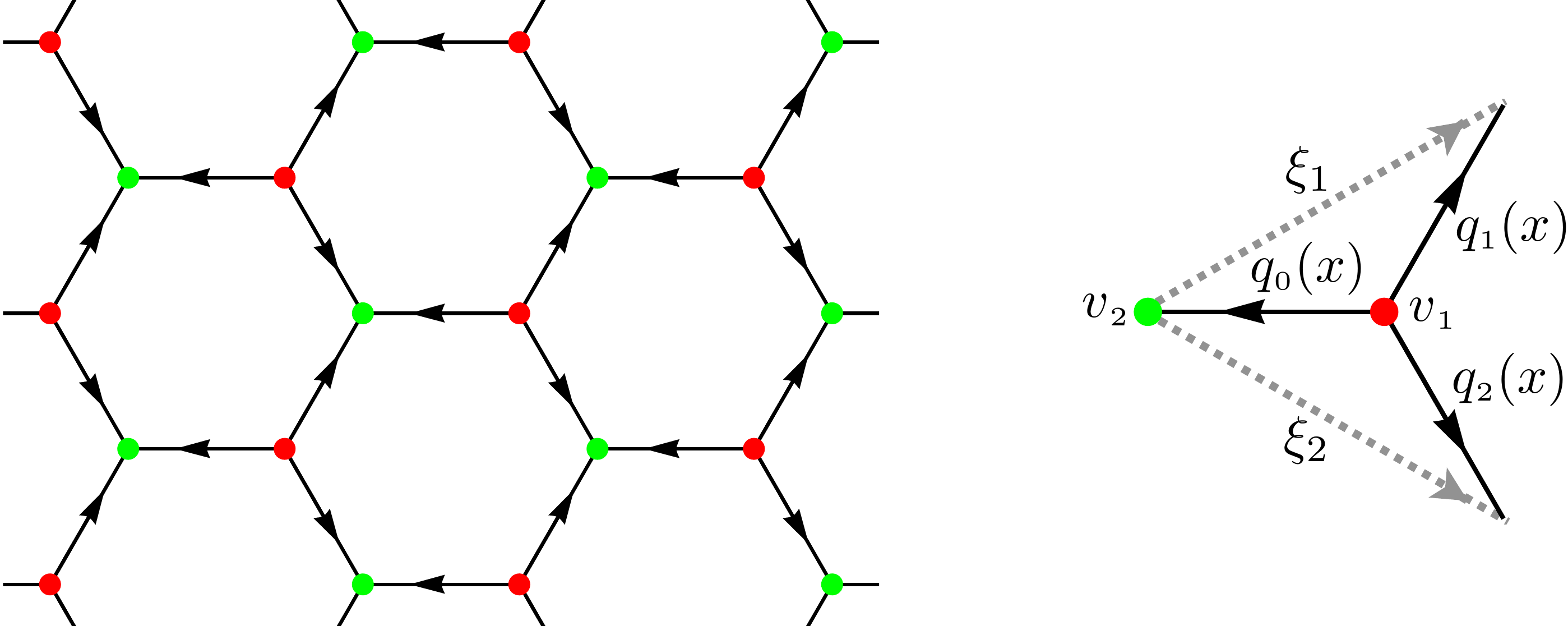}}
}
\caption{\small Single-layer graphene $\Go$ and its fundamental domain.  The arrows on the edges indicate the direction of the $x$-interval $[0,1]$ in the parameterization of the edges.  The vectors $\xi_1$ and $\xi_2$ generate the periodic shifts.}
\label{fig:GrapheneS}
\end{figure}

Let $\xi_1$ and $\xi_2$, as illustrated in Fig.\,\ref{fig:GrapheneS}, be generators of the periodicity in the sense that the action of $(n_1,n_2)\in\ZZ^2$ on $\Go$ shifts the graph along the vector $n_1\xi_1 + n_2\xi_2$ in the plane so that it falls exactly into itself.  The components of the vector $(z_1,z_2)$, the Floquet multipliers, are the eigenvalues of the shifts by $\xi_1$ and $\xi_2$ corresponding to a Floquet mode.  The spectral matrix (\ref{Ahat2}) of this quantum graph at energy $\lambda$ is
\begin{equation}
  \hat {\mathring A}_\lambda(z_1,z_2) \;=\;
  \renewcommand{\arraystretch}{1.2}
\left[
\begin{array}{cc}
  b_1(\lambda) & w(z,\lambda) \\
  \vspace{-2ex} \\
  w(z^{-1},\lambda) & b_2(\lambda)
\end{array}
\right],
\end{equation}
\begin{equation}
  b_1(\lambda) \;=\; -\frac{c_0(\lambda)}{s_0(\lambda)}-\frac{c_1(\lambda)}{s_1(\lambda)}-\frac{c_2(\lambda)}{s_2(\lambda)}-\alpha_1\,,
  \qquad
  b_2(\lambda) \;=\; -\frac{s'_0(\lambda)}{s_0(\lambda)}-\frac{s'_1(\lambda)}{s_1(\lambda)}-\frac{s'_2(\lambda)}{s_2(\lambda)}-\alpha_2\,,
\end{equation}
\begin{equation}\label{w}
  w(z,\lambda) \;=\; \frac{1}{s_0(\lambda)}+\frac{z_1}{s_1(\lambda)}+\frac{z_2}{s_2(\lambda)}\,.
\end{equation}
This is the function $w(z,\lambda)$ in~(\ref{wzlambda}).
Notice that $w(z,\lambda)$ depends only on the potentials $q_i(x)$ through their Dirichlet spectrum since only the functions $s_i(\lambda)$ appear in the definition of~$w(z,\lambda)$.
The dispersion function for $(\Go,\Ao)$ is
\begin{equation}\label{GrapheneD}
  D(z_1,z_2,\lambda) \;=\; \det \hat {\mathring A}_\lambda(z_1,z_2)
  \;=\; b_1(\lambda) b_2(\lambda) - \zeta(z,\lambda),
\end{equation}
in which
\begin{equation}
  \zeta(z,\lambda) \;=\; w(z,\lambda)w(z^{-1},\lambda).
\end{equation}
Observe that the dispersion functions of two different single-layer sheets of graphene have the same $\zeta(z,\lambda)$ exactly when corresponding edges are isospectral, because knowing the Dirichlet spectrum of a potential is equivalent to knowing its $s(\lambda)$ function~\cite[Ch.\,2 Theorem~5]{PoschelTrubowitz1987}.

All three edges in a period of a single layer are isospectral exactly when $s_0(\lambda)=s_1(\lambda)=s_2(\lambda)$, and in this case $\zeta(z,\lambda)$ separates~as
\begin{equation}
  \zeta(z,\lambda) \;=\; s_0(\lambda)^{-2} G(z_1,z_2),
\end{equation}
in which
\begin{equation}
  G(z_1,z_2) \;=\; (1+z_1+z_2)(1+z_1^{-1}+z_2^{-1}).
\end{equation}
The Fermi surface of a single layer at energy $\lambda$ is given by $D(z_1,z_2,\lambda)=0$, which reduces to
\begin{equation}\label{graphenedispersion}
  \Delta(\lambda) \;:=\; s_0(\lambda)^2\, b_1(\lambda)b_2(\lambda) \;=\; G(z_1,z_2).
\end{equation}
We call $\Delta(\lambda)$ the ``characteristic function" for this single-layer graphene model.

For $(z_1,z_2)=(e^{ik_1},e^{ik_2})$ on the torus $\TT^2$,
\begin{equation}
\begin{split}
   \tilde G(k_1,k_2) \;:=\; G(e^{ik_1},e^{ik_2}) &\;=\; \big| 1 + e^{ik_1} + e^{ik_2} \big|^2 \\
       &\;=\; 1 + 8\cos\frac{k_2-k_1}{2}\cos\frac{k_1}{2}\cos\frac{k_2}{2},
\end{split}
\end{equation}
and this has range $[0,9]$ as a function of real $k_1$ and~$k_2$, with its minima occuring at $\pm(2\pi/3,-2\pi/3)$~\cite[Lemma~3.3]{KuchmentPost2007}.  Thus the bands of this graphene model are the real $\lambda$-intervals over which $\Delta(\lambda)$ lies in~$[0,9]$.

Single-layer quantum-graph graphene sheets and tubes, with a common symmetric potential $q_0(x)$ on all edges, are treated in detail in~\cite{KuchmentPost2007}.  In this case, $b_1(\lambda)=b_2(\lambda)$ and the spectrum of the sheet is identical to that of the periodic Hill operator with potential $q_0(x)$ on a period.  In contrast to the Hill operator, the dispersion relation exhibits conical singularities, one for each energy $\lambda$ where $\Delta(\lambda)=0$.  Fig.~\ref{fig:Graphene-D1} shows a graph of~$\Delta(\lambda)$.
Conical singularities are discussed in section~\ref{subsec:conical}.

%%%%%%%%%%%%%%%%%%%%%%%%%%%%%%%%%%%%%%%%%%%%%%%%%
\subsection{Shifting and rotating}

We adopt terminology on shifted layers of graphene that is used in the literature.
The hexagonal graphene structure is invariant under translation by the sum $\xi_1+\xi_2$ of the two elementary shift vectors, as illustrated in Fig.\,\ref{fig:GrapheneSR}.  The shift by $(\xi_1+\xi_2)/3$ (dashed blue) places vertex~$v_2$ onto vertex~$v_1$ and places vertex~$v_1$ onto the center of the hexagon; this will be called the B-shift.  The shift by $2(\xi_1+\xi_2)/3$ (or $-(\xi_1+\xi_2)/3$, dotted orange) places $v_1$ onto $v_2$ and $v_2$ onto the center of the hexagon; this will be called the C-shift.  The unshifted graph is called the A-shift.

By rotating the graphene structure by $\pi$ about the center of an edge, the potentials reverse direction.  This is illustrated on the right of Fig.\,\ref{fig:GrapheneSR}, in which rotation is about the edge labeled~$0$.  Each labeled oriented edge corresponds to a potential $q_i(x)$, with the parameter $x$ increasing in the direction of the arrow.  The labels $0$,~$1$,~$2$ are preserved under rotation, but their orientations are reversed.  Equivalently, rotation effects the change $q_i(x) \mapsto q_i(1-x)$ of the potentials.
The rotation also switches the Robin conditions on the two vertices of a period.

Denote a single layer by $(\Lo,\Ao)$ and its $180^\circ$ rotation by $(\Lo,\Ao_\pi)$.
The potentials $q_i(x)$ and $q_i(1-x)$ have the same Dirichlet spectrum, which coincides with the roots of the function~$s_i(\lambda)$.  Therefore the function $w(z,\lambda)$ in (\ref{w}) is the same for both quantum graphs and their dispersion functions are polynomials in the same composite Floquet variable~$\zeta(z,\lambda)=w(z,\lambda)w(z^{-1},\lambda)$.

\begin{figure}[h]
\centerline{
\raisebox{13pt}{\scalebox{0.55}{\includegraphics{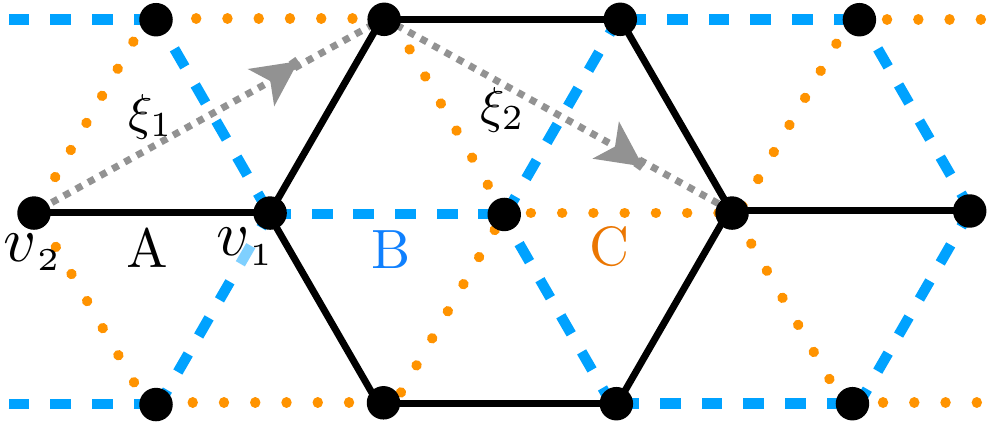}}}
\hspace{4em}
\raisebox{0pt}{\scalebox{0.36}{\includegraphics{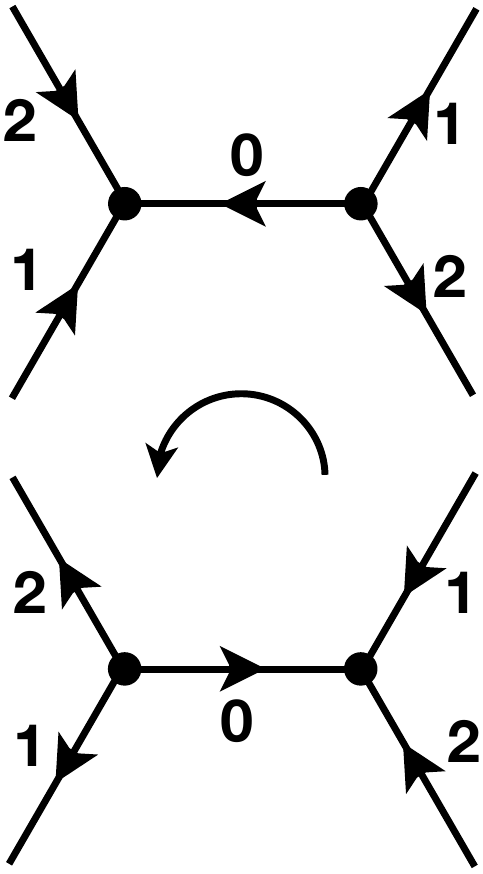}}}
}
\caption{\small Left: A- B- and C-shifts of graphene are illustrated in solid black, dashed blue, and dotted orange, as described in the text.  Right:  Rotating graphene by $180^\circ$ reverses the orientation of the potentials but preserves their Dirichlet spectra.}
\label{fig:GrapheneSR}
\end{figure}

\begin{figure}[h]
\centerline{
\scalebox{0.35}{\includegraphics{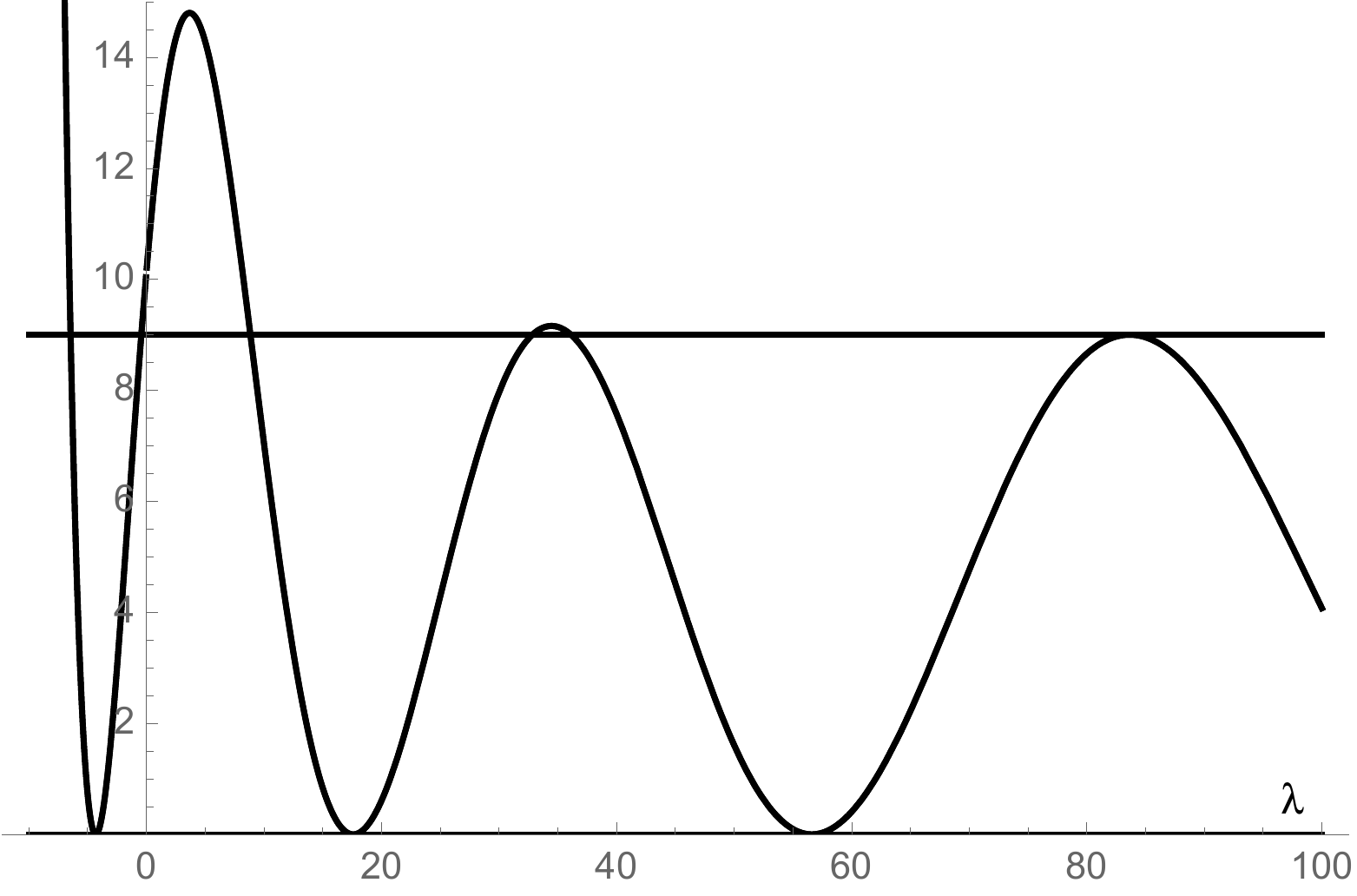}}\;\raisebox{5pt}{\small $\lambda$}
}
\caption{\small Graph of the characteristic function $\Delta(\lambda)$ of single-layer graphene, showing the first three spectral bands and the first three gaps.  The bands are the $\lambda$-intervals for which $\Delta(\lambda)\in[0,9]$, which is the range of the function $\tilde G(k_1,k_2)$.  The points where $\Delta(\lambda)=0$ correspond to conical singularities of the dispersion relation $D(e^{ik_1},e^{ik_2},\lambda)=0$, which occur inside the bands, as discussed in section~\ref{subsec:conical} and~\cite{KuchmentPost2007}}
\label{fig:Graphene-D1}
\end{figure}

%%%%%%%%%%%%%%%%%%%%%%%%%%%%%%%%%%%%%%%%%%%%%%%%%
\subsection{AA-stacking and rotation}\label{subsec:AA}

In AA-stacked graphene, each layer is stacked directly over the previous and each pair of vertically successive vertices is connected by an edge, as in Fig.\,\ref{fig:GrapheneAA}.  As a type-2 $n$-layer graph, the red vertices in a given period, together with the $n\!-\!1$ edges connecting them, form the connector graph $(\Sigma_1,B_1)$, and the green vertices and the edges connecting them form~$(\Sigma_2,B_2)$.
The hypotheses of Theorem~\ref{thm:type2} allow the potentials $q(x)$ on any pair of vertically aligned edges on two different layers to differ as long as the operators $-d^2/dx^2 + q_e(x)$ possess the same Dirichlet spectrum.  The theorem then guarantees that the Fermi surface of the layered structure is reducible with $n$ components.

A particular instance of AA-stacked graphene satisfying the hypotheses of the theorem is constructed from copies of a given single layer and its rotations about the center of an edge.
Let a single layer $(\Lambda,A_0)$ with arbitrary potentials on the three edges of a period and arbitrary Robin parameters on the two vertices be given.  Rotation of this graph by $180^\circ$ about the center of an edge, as described in the previous section and illustrated in Figs.~\ref{fig:GrapheneSR},\ref{fig:GrapheneAA} (right), results in a new layer of graphene $(\Lambda,A_\pi)$ with the same underlying graph~$\Lambda$ but with the potentials oriented in the opposite direction and the Robin parameters at the two vertices switched.

Thus Theorem~\ref{thm:type2} applies to an $n$-layer stack, with each layer being either $(\Lo,\Ao)$ or $(\Lo,\Ao_\pi)$, in any order, stacked in the AA sense.  The Fermi surface of this $n$-layer graphene has $n$ components.  According to section~\ref{sec:iterative2}, equation~(\ref{D2}), the relation $D(z,\lambda)=0$ reduces to $n$ components
\begin{equation}\label{mui}
  \mu_i(\lambda) \;=\; G(z_1,z_2)
  \qquad
  i=1,\dots,n,
\end{equation}
in which $\mu_i$ are the eigenvalues of the ``characteristic matrix"
\begin{equation}\label{Deltamatrix}
  \Delta(\lambda) \;=\; s_0(\lambda)^2\,\tilde B_1(\lambda) \tilde B_2(\lambda),
\end{equation}
which generalizes the characteristic function~(\ref{graphenedispersion}) by the same name for the single layer.

\begin{figure}[h]
\centerline{
\scalebox{0.35}{\includegraphics{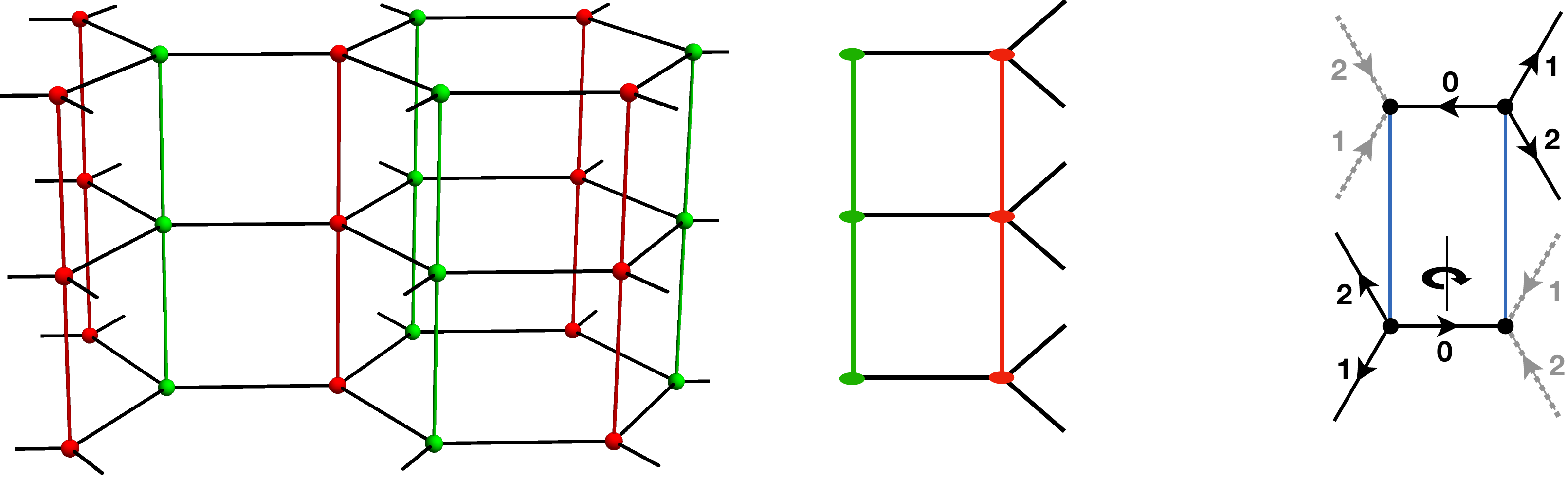}}
}
\caption{\small AA-stacked graphene in three layers and a fundamental domain thereof.  If the potentials on corresponding edges on different layers have the same Dirichlet spectrum, then the Fermi surface for the multi-layer graph is reducible.  This occurs, in particular, when a layer is rotated by $\pi$ about the center of an edge.}
\label{fig:GrapheneAA}
\end{figure}

\smallskip
{\bfseries Examples.}  (Fig.~\ref{fig:GrapheneAA-D2} and~\ref{fig:GrapheneAA-D3}) Let the layers be identical with identical potential $q_0(x)$ on all three edges of a period.  We take $q_0(x)$ to be symmetric about $x=1/2$ so that the DtN map for the edge is independent of the direction.  We choose
\begin{equation}\label{q0}
  q_0(x) \;=\; -16\,\cchi_{[1/3,2/3]}(x)
\end{equation}
($\cchi_Y(x)$ is the characteristic function of the set $Y\subset\RR$) so that the DtN map is explicitly computable and so that the spectrum of the single layer does have gaps (because $q_0(x)$ is not constant; see~\cite{KuchmentPost2007}).

On all of the connector edges, which are all of length~$1$, we take the potential $q_\conn(x)$ to be either $0$ or $q_0(x)$~or
\begin{equation}\label{qc}
  q_\conn(x) \;=\; -10\,\cchi_{[1/2,1]}(x),
\end{equation}
which is not symmetric about the center.  Fig.~\ref{fig:GrapheneAA-D2} and~\ref{fig:GrapheneAA-D3} show graphs of $\mu_i(\mu)$ for bi-layer and tri-layer graphene.  Each eigenvalue contributes a sequence of bands and gaps to the spectrum of the multi-layer graph---the bands for the $i^\text{th}$ sequence are the $\lambda$-intervals for which $\mu_i(\lambda)\in[0,9]$.  When the Dirichlet spectral function $s(\lambda)$ on the connecting edges is different from that of the layers, new thin bands are introduced.  Conical singularities, or Dirac cones, are discussed below.  These are characteristic features of single-layer graphene, and in special cases of AA-stacking, they persist, according to Proposition~\ref{prop:AAconical}.

\begin{figure}[h]
\centerline{
\scalebox{0.42}{\includegraphics{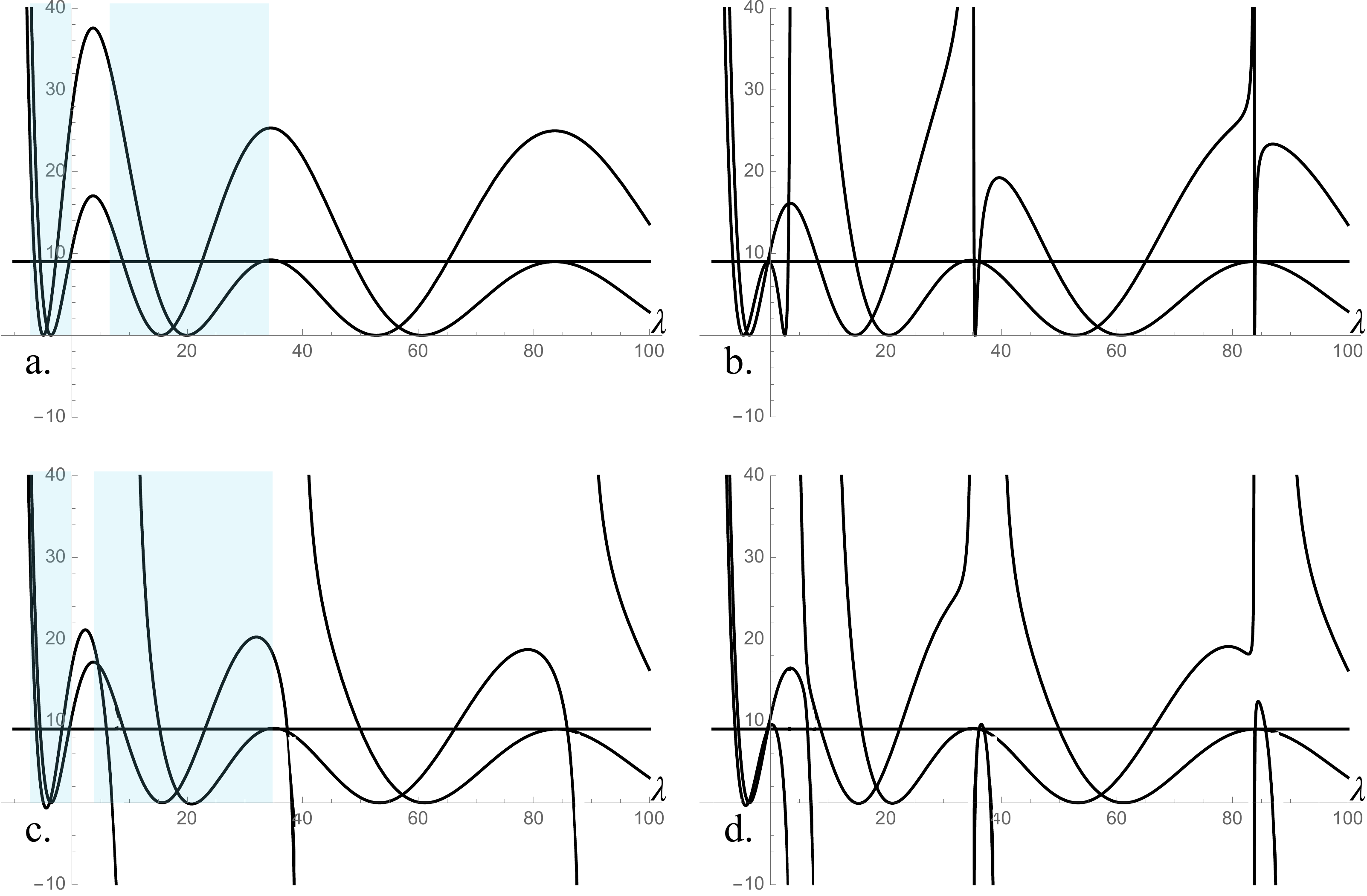}}
}
\caption{\small For double-layer AA-stacked graphene, the two eigenvalues $\mu_i(\lambda)$ ($i=1,2$) of the characteristic matrix $\Delta(\lambda)$ give two sets of spectral bands and gaps.  The bands are the $\lambda$-intervals for which $\mu_i(\lambda)\in[0,9]$.
The dispersion relation for the shaded $\lambda$-intervals are shown in Fig.\,\ref{fig:GrapheneAA-D2-3D}.
{\bfseries a.} The connecting edges have the same potential as those of the layers~(\ref{q0}).  {\bfseries b.} The potentials~(\ref{qc}) of the two connecting edges are equal to each other but different from that of the layers~(\ref{q0}). This creates additional thin bands (within each of the sets of spectral bands), which have conical singularities of their own.
{\bfseries c.} The potentials of the two connecting edges are different from each other ($q(x)=0$ and~\ref{q0}).  This destroys conical singularities and introduces additional thin gaps in their place.  Additionally, new thin bands are introduced just below the vertical asymptotes. {\bfseries d.} The potentials of the two connecting edges are different from each other ($q(x)=0$ and~\ref{qc}).}
\label{fig:GrapheneAA-D2}
\end{figure}

\begin{figure}[h]
\centerline{
\scalebox{0.4}{\includegraphics{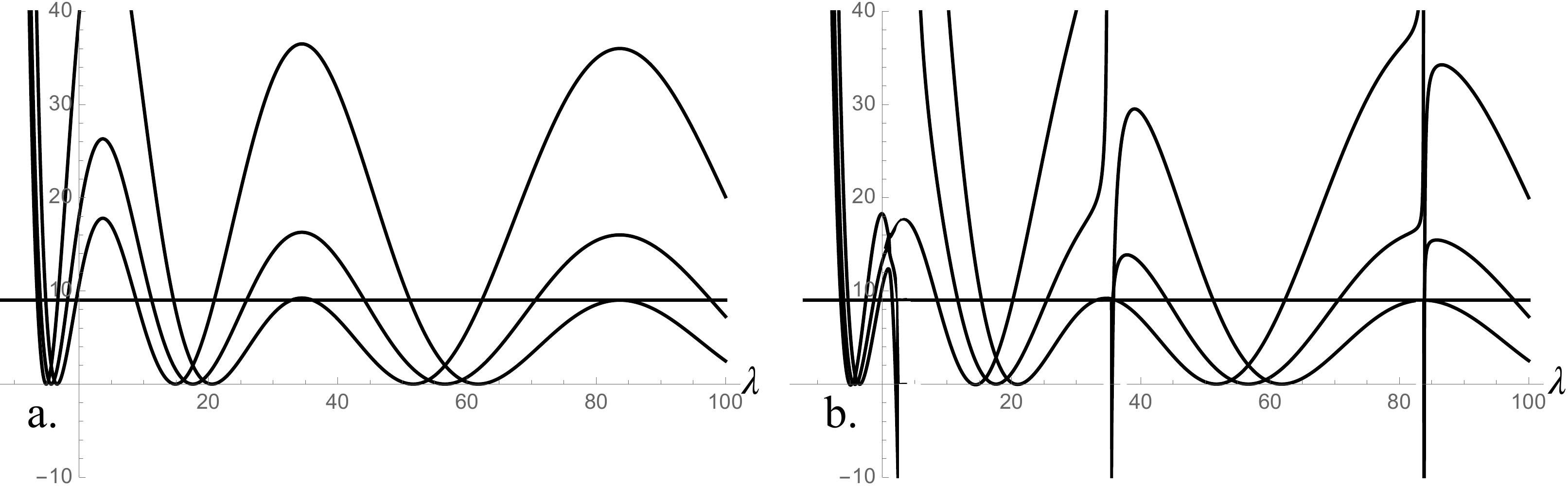}}
}
\caption{\small The three eigenvalues $\mu_i(\lambda)$ ($i=1,2,3$) of the characteristic matrix $\Delta(\lambda)$ for triple-layer AAA-stacked graphene, giving three sets of spectral bands and gaps.  The bands are the $\lambda$-intervals for which $\mu_i(\lambda)\in[0,9]$.  {\bfseries a.}~The connecting edges have the same potential as those of the layers.  {\bfseries b.} The connecting edges have potential~(\ref{q0}) at vertex $v_1$ and potential~(\ref{qc}) at vertex $v_2$; this creates additional thin gaps and destroys Dirac cones (the minima are slightly below~$0$).}
\label{fig:GrapheneAA-D3}
\end{figure}

%%%%%%%%%%%%%%%%%%%%%%%%%%%%%%%%%%%%%%%%%%%%%%%%%
\subsection{AB-stacking}\label{subsec:AB}

In AB-stacked (or ABABA, ABBA, {\itshape etc.}) graphene, the layers are in the A-shift or the B-shift and are coupled by a single edge per period.  We allow any number of layers with the A- and B-shifts arranged in arbitrary order.
Fig.\,\ref{fig:GrapheneAB} illustrates three layers with alternating shifts.

In each layer, we allow both $(\Lo,\Ao)$ and $(\Lo,\Ao_\pi)$ or any potentials $q_i(x)$ ($i=1,2,3$) as long as, for each~$i$, the Dirichlet spectra are invariant across layers.
As noted above, this guarantees that the function $\zeta(z,\lambda)$ is independent of the layer since it depends only on the Dirichlet spectral functions $s_i(\lambda)$, which are equivalent to the Dirichlet spectra of the potentials~$q_i(x)$.
Note that isospectrality (which is explicitly required for type~2) arises for graphene in type-1 stacking.
Thus the dispersion function of each layer is of the form~(\ref{GrapheneD}) with different $b_i(\lambda)$ but the same~$\zeta(z,\lambda)$.  
In any period of this layered structure, $n$ vertices, one per layer, are aligned along a vertical line, and these are connected by edges.  These vertices serve as vertices of separation of the individual layers.  Thus Theorem~\ref{thm:type1} on type-1 multi-layer graphs applies.

\begin{figure}[h]
\centerline{
\scalebox{0.4}{\includegraphics{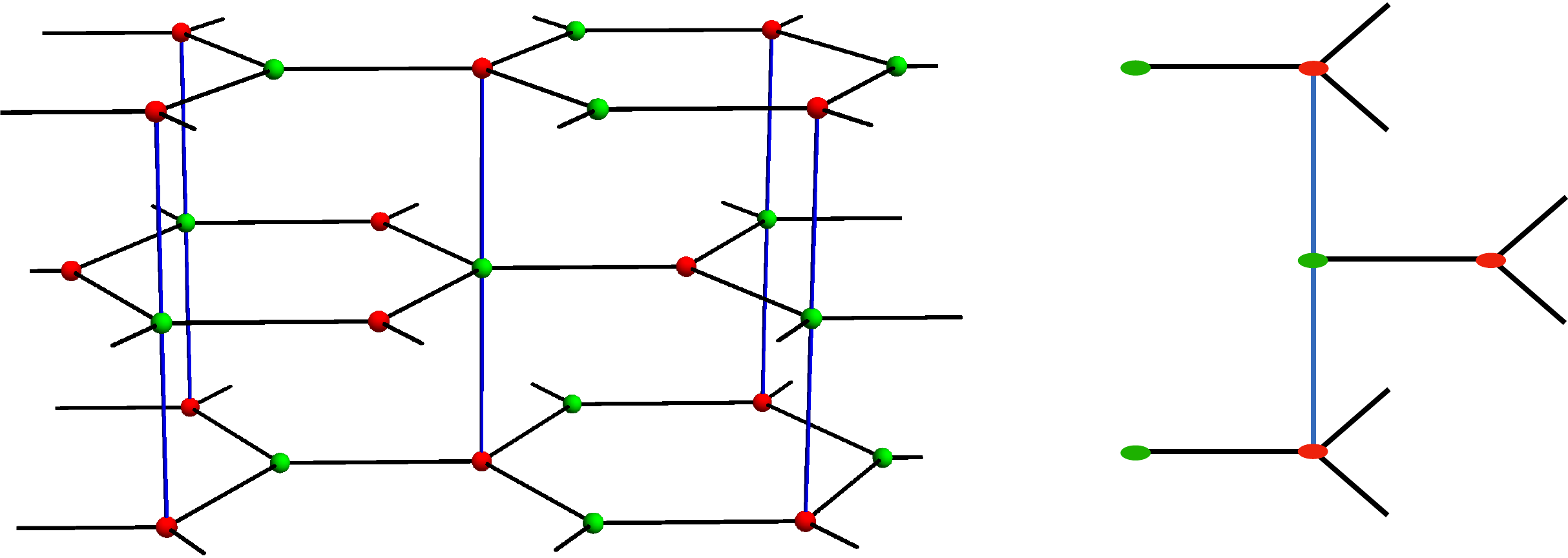}}
}
\caption{\small AB-stacked graphene in three layers (ABA) and a fundamental domain thereof.}
\label{fig:GrapheneAB}
\end{figure}

\begin{figure}[h]
\centerline{
\scalebox{0.4}{\includegraphics{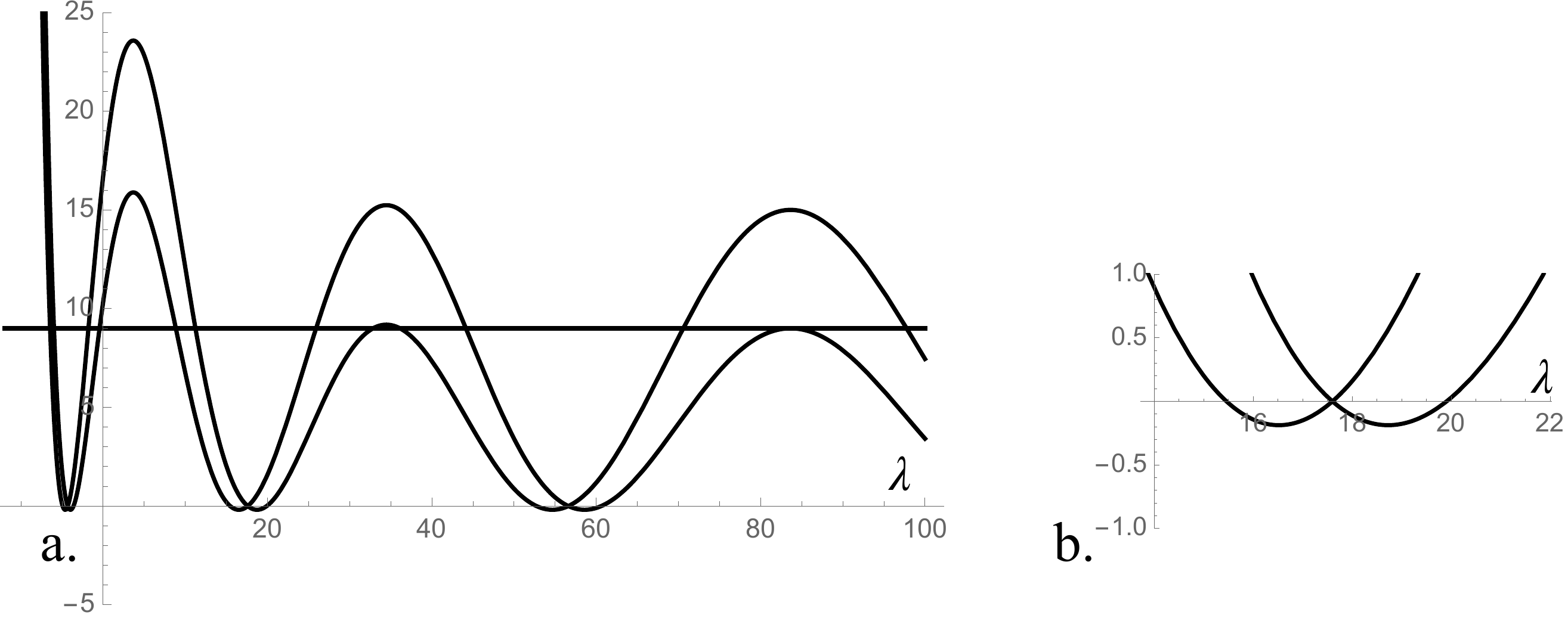}}
}
\caption{\small For double-layer AB-stacked graphene, these are graphs of $s(\lambda)^2\mu_i(\lambda)$ ($i=1,2$), where $\mu_i(\lambda)$ are the two roots of $P(\zeta(z,\lambda),\lambda)$, as in Theorem~\ref{thm:type2}.  Each root gives a set of spectral bands and gaps.  The bands are the $\lambda$-intervals for which $\mu_i(\lambda)\in[0,9]$.  {\bfseries a.}~The connecting edge has the same potential as those of the layers.
{\bfseries b.}~A close view near $\lambda=20$ shows that the graphs of $\mu_i(\lambda)$ cross the horizontal axis, and thus Dirac cones are not present.}
\label{fig:GrapheneAB-D2}
\end{figure}

\smallskip
{\bfseries Examples.}  (Fig.\,\ref{fig:GrapheneAB-D2}) As in the previous subsection, let the layers be identical with identical symmetric potential $q_0(x)$ on all edges.  Fig.~\ref{fig:GrapheneAB-D2} shows the graphs of the roots $\mu_i(\lambda)$ of the polynomial $P(\zeta,\lambda)$---see Theorems~\ref{thm:type1} and~\ref{thm:reducible1}.  Conical singularities of the dispersion relation are discussed below.

%%%%%%%%%%%%%%%%%%%%%%%%%%%%%%%%%%%%%%%%%%%%%%%%%
\subsection{ABC-stacking}\label{subsec:ABC}

In ABC-stacked graphene, all three shifts are stacked, as illustrated in Fig.\,\ref{fig:GrapheneABC}.  The number of components of the Fermi surface of the ABC-stacked structure is equal to the number of layers.  We leave the details of how to use Theorems~\ref{thm:type1} and~\ref{thm:type2} to prove this to the reader.  The arguments are similar to those described for the more general mixed stacking below.

\begin{figure}[h]
\centerline{
\scalebox{0.3}{\includegraphics{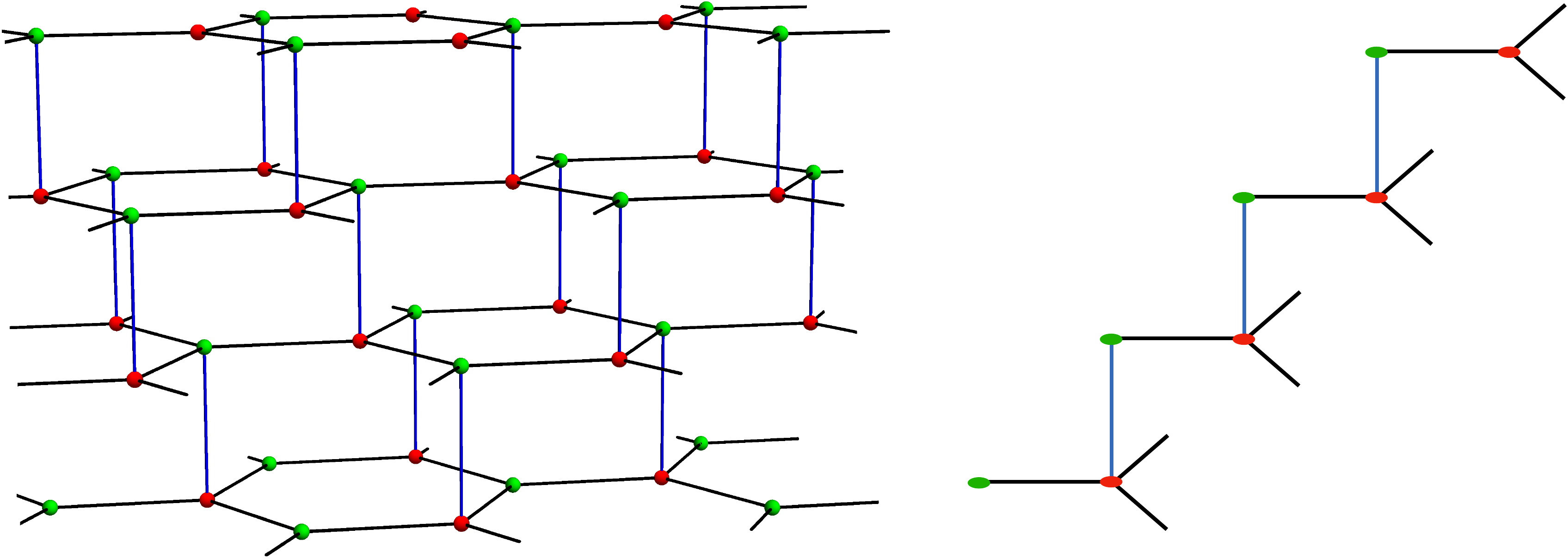}}
}
\caption{\small ABC-stacked graphene in four layers and a fundamental domain thereof.}
\label{fig:GrapheneABC}
\end{figure}

%%%%%%%%%%%%%%%%%%%%%%%%%%%%%%%%%%%%%%%%%%%%%%%%%
\subsection{Mixed stacking}\label{subsec:mixed}

Several graphene sheets can be stacked with arbitrary shifts (A,B,C) to obtain a mixed-stacking multi-layer sheet of graphene with reducible Fermi surface, as long as all the individual layers have dispersion function that is a polynomial in the same composite Floquet variable~$\zeta(z,\lambda)$.
As discussed above, this occurs when all vertically aligned edges are Dirichlet-isospectral.  Fig.\,\ref{fig:GrapheneM} depicts five layers stacked with mixed shifts.  The dispersion function is a polynomial in $\zeta(z,\lambda)$ whose degree is the number of single sheets of graphene in the stack.  The proof of this uses iterated application of the theorems on type-1 and type-2 multi-layer constructions and section~\ref{subsec:several}, as described next.

Let $\Sigma_1$ and $\Sigma_2$ be $n$-layer and $m$-layer AA-stacked graphene quantum graphs (type\,2).  Let the individual layers of both have dispersion functions that are polynomials in a common $\zeta(z,\lambda)$.
Let $u_1$ and $u_n$ be corresponding (vertically aligned) vertices in the first and $n$-th layers of $\Sigma_1$, and let $v_1$ and $v_m$ be corresponding vertices in the first and $m$-th layers of $\Sigma_2$.

Observe that $\Sigma_1^{u_n}$ has dispersion function that is also a function of the same $\zeta(z,\lambda)$.  This is because a single graphene layer is separable at any vertex, and thus $\Sigma_1^{u_n}$ can be viewed as a type-2 $(n\!-\!1)$-layer graph with connectors that consist of the edges between the $n\!-\!1$ layers plus decorations.  The same is true of~$\Sigma_2^{v_0}$.  Therefore $\Sigma_1$ and $\Sigma_2$ can be coupled by an edge between $u_n$ and $v_0$ according to a two-layer type-1 construction, resulting in a quantum graph $\G$ with dispersion function that is a polynomial in $\zeta(z,\lambda)$.

This construction could just as well be carried out using $\Sigma_2^{v_m}$ in place of $\Sigma_2$, resulting in $\G^{v_m}$, whose dispersion function is a polynomial in $\zeta(z,\lambda)$.  Now yet another AA-stacked multi-layer graphene construction $\Sigma_3$ with the same $\zeta(z,\lambda)$ can be attached to $\G$, and so on.
These arguments need to modified somewhat if any of the AA-stacked sections $\Sigma_i$ consists of only one layer, as in ABC-stacked graphene.

\begin{figure}[h]
\centerline{
\scalebox{0.3}{\includegraphics{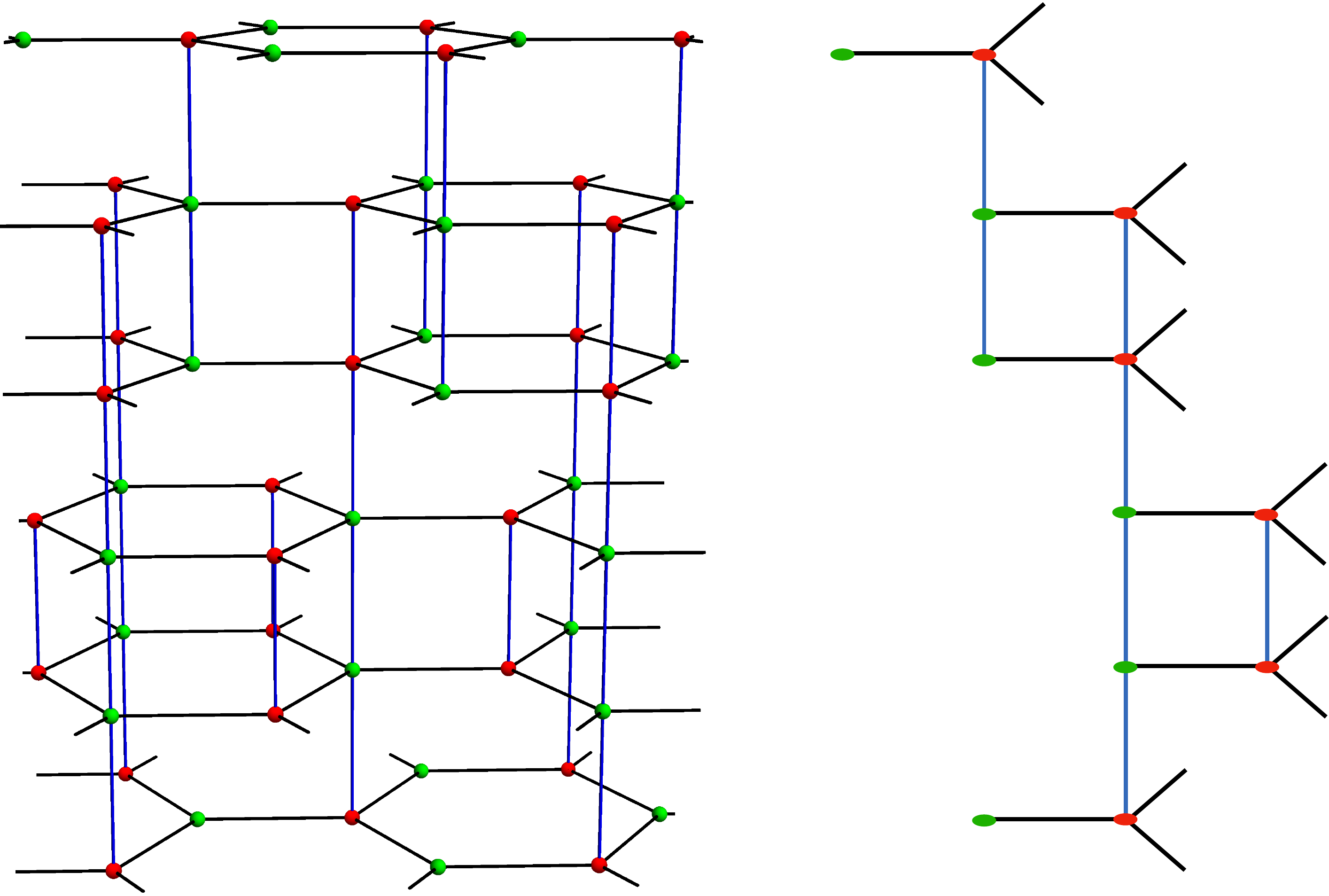}}
}
\caption{\small Five layers of graphene in mixed stacking, and a fundamental domain thereof.}
\label{fig:GrapheneM}
\end{figure}

%%%%%%%%%%%%%%%%%%%%%%%%%%%%%%%%%%%%%%%%%%%%%%%%%
\subsection{Conical singularities}\label{subsec:conical}

Single-layer graphene is famous for the ``Dirac cone" feature of its dispersion relation $D(e^{ik_1},e^{ik_2},\lambda)=0$.  The Dirac cone is a conical singularity in (a branch of) the dispersion relation $D(\bar k,\lambda)=0$ ($\bar k=(k_1,k_2)$), where it has the approximate form
\begin{equation}
  (\lambda-\lambda_0)^2 \;\approx\; c\, \big| \bar k - \bar k_0 \big|^2
\end{equation}
($c>0$) to leading order for $(\bar k,\lambda)$ near $(\bar k_0,\lambda_0)$.
There is a large amount of literature on this, for example~\cite{AbergelApalkovBerashevic2010a,Castro-NetGuineaPeres2009a,CastroNovoselovMorozov2007a,McCann2006a,McCannAbergelFalko2007a,PartoensPeeters2006a}.  For Schr\"odinger operators in $\RR^2$ with very general honeycomb-type potentials, the existence of circular Dirac cones and their stability is proved by perturbation methods in~\cite{FeffermanWeinstein2012}.  In~\cite{BerkolaikoComech2018}, the authors provide a treatment for more general operators with hexagonal periodicity that is highlights the representations of various symmetry groups.
 
For single-layer graphene with equal and symmetric potentials on all edges, conical singularities occur at energies~$\lambda$ for which $\Delta(\lambda)=0$~\cite{KuchmentPost2007}.  This is because the dispersion relation is $\Delta(\lambda)=\tilde G(k_1,k_2)$ and
both $\Delta(\lambda)$ and $\tilde G(k_1,k_2)$ have $0$ as nondegenerate minima (here, $\tilde G(k_1,k_2):=G(e^{ik_1},e^{ik_2})$).  The quasimomenta at these points are $\pm(2\pi/3,-2\pi/3)$.

When $n$ layers are stacked in the AA sense, the $n$ branches of the dispersion relation are~$\mu_i(\lambda)=\tilde G(k_1,k_2)$ (\ref{mui}).  When the connecting graph (sequence of connecting edges in this case) for the red vertices is the same as that for the green vertices, then $B_1(\lambda)=B_2(\lambda)$.  The symmetry of the potentials about the centers of the edges implies $c(\lambda)=s'(\lambda)$; this, together with equal Robin parameters for each layer results in $\mathbf{b}_1(\lambda)=\mathbf{b}_2(\lambda)$.
Thus $\tilde B_1(\lambda)=\tilde B_2(\lambda)$, and the $\mu_i(\lambda)$ are eigenvalues of the positive matrix~$\Delta(\lambda)$ (\ref{Deltamatrix}). Therefore, in each spectral band, $\mu_i(\lambda)$ reaches its minimal value of $0$, and a conical singularity occurs (as long as the minima are nondegenerate), as shown in Fig.~\ref{fig:GrapheneAA-D2}(a,b) and Fig.~\ref{fig:GrapheneAA-D3}(a)---this is stated in the next proposition.
\begin{proposition}[Condition for conical singularities]\label{prop:AAconical}
For an AA-stacked multi-layer graphene structure: if
(1) the potentials on all the edges of all layers (black edges in Fig.~\ref{fig:GrapheneAA}) are identical and symmetric about the center of the edge; (2) in each layer, the two Robin parameters are equal; and (3) the potentials connecting the green vertices (green edges in Fig.~\ref{fig:GrapheneAA}) are the same as those connecting the red vertices (red edges in Fig.~\ref{fig:GrapheneAA}), then the dispersion relation has a conical singularity at each energy $\lambda$ for which an eigenvalue $\mu(\lambda)$ of $\Delta(\lambda)$ is equal to zero and the second derivative of $\mu(\lambda)$ is nonzero.
\end{proposition}

This proposition can also be obtained from~\cite[Theorem\,2.4]{BerkolaikoComech2018}, as the conditions imply symmetry under rotation by $\pi/3$, inversion, and reflection, in the plane of the layers.

Incidentally, condition (1) in Proposition~\ref{prop:AAconical} apparently cannot be relaxed.  The proof relies on two conditions on the potentials of the layers:  They must all be isospectral and they must be symmetric.  It is known (see~\cite{PoschelTrubowitz1987}) that the Dirichlet spectrum completely determines the potential $q(x)$ within the class of symmetric potentials.  Therefore all potentials on all the layers must be equal.

When the green and red vertices are connected differently, $\tilde B_1(\lambda)\tilde B_2(\lambda)$ is not necessarily a positive matrix, and indeed the $\mu_i(\lambda)$ cross~$0$ linearly, at which points there are nondegenerate band edges.  This is illustrated in Fig.~\ref{fig:GrapheneAA-D2}(c,d) and Fig.~\ref{fig:GrapheneAA-D3}(b).  Dispersion relations in $(k_1,k_2,\lambda)$-space are shown in Fig.\,\ref{fig:GrapheneAA-D2-3D}.

For AB-stacked graphene, the numerical computation in Fig.~\ref{fig:GrapheneAB-D2}(a,b) shows $\mu_i(\lambda)$ crossing $0$ linearly at each point where it vanishes, and thus each of these points corresponds to a nondegenerate band edge (and not a conical singularity).

\begin{figure}[h]
\centerline{
\scalebox{0.4}{\includegraphics{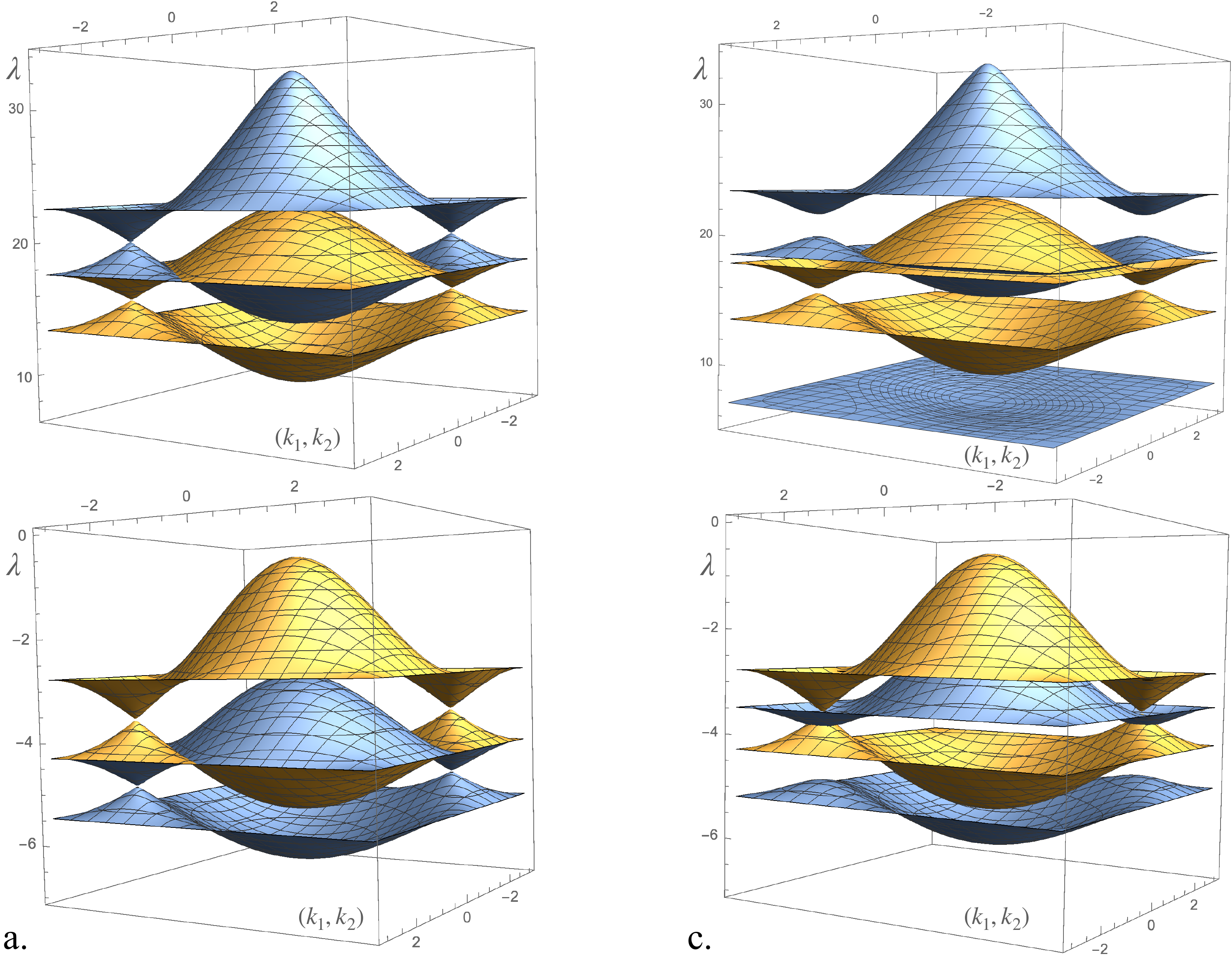}}
}
\caption{\small The dispersion relations for AA-stacked graphene from examples (a) and (c) in Fig.~\ref{fig:GrapheneAA-D2}, shown for energies $\lambda$ within the shaded intervals in Fig.~\ref{fig:GrapheneAA-D2}(a,c).  The two sets of bands coming from the eigenvalues $\mu_i(\lambda)$ ($i=1,2$) are in blue and yellow.  {\bfseries a.} The two connecting edges have the same potential.  Each of the energies $\lambda$ where the $\mu_i(\lambda)$ hit zero in Fig.~\ref{fig:GrapheneAA-D2}(a) exhibits a conical singularity at two different quasimomenta $(k_1,k_2)$.  {\bfseries c.} The two connecting edges have different potentials.  All of the conical points open into gaps.  In the upper interval, an additional thin band appears.}
\label{fig:GrapheneAA-D2-3D}
\end{figure}

%%%%%%%%%%%%%%%%%%%%%%%%%%%%%%%%%%%%%%%%%%%%%%%%%
%%%%%%%%%%%%%%%%%%%%%%%%%%%%%%%%%%%%%%%%%%%%%%%%%
\section{Some irreducible Fermi surfaces}\label{subsec:irreducible}

We give two examples of multi-layer quantum graphs that are not of type 1 or type 2 and whose Fermi surface is not reducible for some open set of energies.

\smallskip
The first example is a bi-layer graph with identical tripartite layers; the single layer is shown in Fig.~\ref{fig:Tripartite}.  Each edge of the single layer has the same symmetric potential.  Two copies of this layer are connected by two edges per period, one edge between green vertices and one between red, and these two edges have potentials from different asymmetry classes as defined in~\cite{Shipman2019} (otherwise the Fermi surface would be reducible by \cite[Theorem\,4]{Shipman2019}).
With the vertices ordered green-red-blue, the spectral matrix for this graph is
\begin{align}
\hat A(z,\lambda) &\;=\;
\renewcommand{\arraystretch}{1.2}
\frac{1}{s(\lambda)}
\left[
  \begin{array}{cccccc}
    -3c(\lambda) & 0 & 1+z_2 & 0 & 1 & 0 \\
    0 & -3c(\lambda) & 0 & 1+z_2 & 0 & 1 \\
    1+z_2^{-1} & 0 & -2c(\lambda)-2s'(\lambda) & 0 & 1+z_1 & 0 \\
    0 & 1+z_2^{-1} & 0 & -2c(\lambda)-2s'(\lambda) & 0 & 1+z_1 \\
    1 & 0 & 1+z_1^{-1} & 0 & -3s'(\lambda) & 0 \\
    0 & 1 & 0 & 1+z_1^{-1} & 0 & -3s'(\lambda)
  \end{array}
\right] \\
 &\;+\; 
\left[
  \begin{array}{cccclr}
  -c_1(\lambda)s_1(\lambda)^{-1} & s_1(\lambda)^{-1} &0&0&0&0 \\
  s_1(\lambda)^{-1} & -s_1'(\lambda)s_1(\lambda)^{-1} &0&0&0&0 \\
  0&0& -c_2(\lambda)s_2(\lambda)^{-1} & s_2(\lambda)^{-1} &0&0 \\
  0&0& s_2(\lambda)^{-1} & -s_2'(\lambda)s_2(\lambda)^{-1} &0&0 \\
  0&0&0&0&0&0 \\
  0&0&0&0&0&0
  \end{array} 
\right].
\end{align}

\begin{figure}[h]
\centerline{
\scalebox{0.27}{\includegraphics{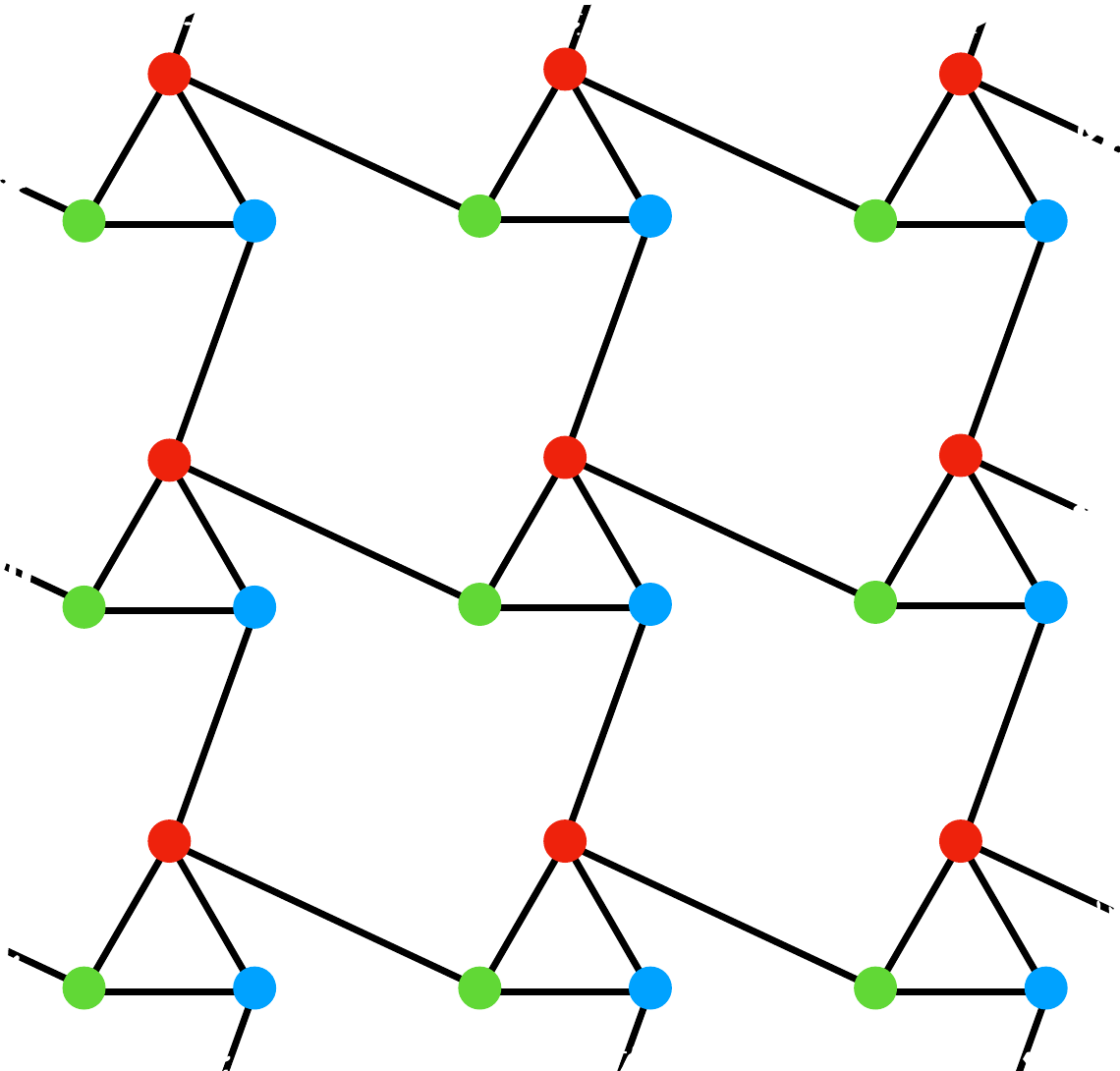}}
}
\caption{\small A tripartite periodic graph.  Two layers of this graph coupled at the green and red vertices with appropriate potentials yields an irreducible Fermi surface for some energies.}
\label{fig:Tripartite}
\end{figure}

We prove, with the help of a computer, that the determinant of this matrix is not factorable as a product of non-monomial Laurent polynomials in $z_1$ and $z_2$.
We use but do not prove here the surjectivity of the Dirichlet-to-Neumann map of the Schr\"odinger operator $-d^2/dx^2 + q(x)$ on an interval:  Given $\lambda\in\mathbb R$ and three real numbers $\alpha$, $\beta$ and $\gamma$, there exists a continuous potential function $q(x)$, such that the corresponding spectral functions of $q(x)$ satisfy $c(\lambda)=\alpha$, $s(\lambda)=\beta$, and $s'(\lambda)=\gamma$. 
By this surjectivity, we can choose the values of $c(\lambda)$, $s(\lambda)$, $s'(\lambda)$, $c_1(\lambda)$, $s_1(\lambda)$, $s_1'(\lambda)$, $c_2(\lambda)$, $s_2(\lambda)$ and $s_2'(\lambda)$ independently, and the determinant becomes a Laurent polynomial in $z_1$ and $z_2$ with real coefficients.
%Then we use the function ``Factor" in Mathematica to determine the factorability, as this function factors polynomials of two or three variables up to degree 100 efficiently (https://reference.wolfram.com/language/guide/PolynomialFactoring.html).
Nonfactorability occurs, for example, for $c=s=s'=s_1=s_1'=c_1=s_2=s_2' = 1$ and $c_2=2$, 
%$c(\lambda) = c_1(\lambda) = c_2(\lambda) = 2$, $s(\lambda) = s_1(\lambda) = s_2(\lambda) = 1$ and $s'(\lambda) = s_1'(\lambda) = s_2'(\lambda) = 1$,
and therefore also for an open set of energies around $\lambda$ for the same potentials.
 
\smallskip
The second example is crossed bi-layer graphene.
The layers are identical, and, within one period, the red vertex of each layer is connected to the green vertex of the other layer.   With $w = 1+z_1^{-1}+z_2^{-1}$ and $w' = 1 + z_1 + z_2$ the spectral matrix for this graph is
\begin{align}
  \hat A(z,\lambda) &\;=\;
\renewcommand{\arraystretch}{1.2}
\frac{1}{s(\lambda)}\left[
\begin{array}{cccc}
  -3c(\lambda) & 0 & w & 0 \\
  0 & -3c(\lambda) & 0 & w \\
  w' & 0 & -3s'(\lambda) & 0 \\
  0 & w' & 0 & -3s'(\lambda) 
\end{array}
\right]  \\
&\;+\;
\renewcommand{\arraystretch}{1.2}
\left[
\begin{array}{cccc}
  -c_1(\lambda)s_1(\lambda)^{-1} &0&0& s_1(\lambda)^{-1} \\
  0& -c_2(\lambda)s_2(\lambda)^{-1} & s_2(\lambda)^{-1} &0 \\
  0& s_2(\lambda)^{-1} & -s'_2(\lambda)s_2(\lambda)^{-1} &0 \\
  s_1(\lambda)^{-1} &0&0& -s'_1(\lambda)s_1(\lambda)^{-1}
\end{array}
\right].
\end{align}
The choice $c=s=s'=s_1=s_1'=c_1=s_2=s_2' = 1$ and $c_2=2$ renders the determinant not factorable.  Incidentally, when the connecting edges have the same potential, the determinant factors, but the factors are not functions of~$ww'$ or any other function of $z_1$ and~$z_2$.

%%%%%%%%%%%%%%%%%%%%%%%%%%%%%%%%%%%%%%%%%%%%%%%%%
%%%%%%%%%%%%%%%%%%%%%%%%%%%%%%%%%%%%%%%%%%%%%%%%%
\section{Appendix:  Moving the poles of the dispersion function}\label{sec:dotted}

This appendix proves Proposition~\ref{prop:poles} in section~\ref{subsec:notation}.  The proof is based on the dotted-graph technique~\cite{KuchmentZhao2019a}.
First consider the simple case of a quantum graph where the underlying graph $E$ consists of two vertices and the edge $e=(v_1,v_2)$ between them, identified with the $x$-interval $[0,L]$.  With the operator $-d^2/dx^2-q(x)$ and Robin parameters $\alpha_1$ at $v_1 (x=0)$ and $\alpha_2$ at $v_2 (x=L)$, we obtain a quantum graph~$(E,Q)$.
Let $\dot E$ be the graph obtained by placing an additional vertex $v$ at the point of $e$ corresponding to $x=\ell\in\;]0,L[$; thus $\dot E$ consists of three vertices and two edges $e_1=(v_1,v)$ and $e_2=(v,v_2)$, with $e_1$ identified with $[0,\ell]$ and $e_2$ identified with $[\ell,L]$.

Restricting the potential $q$ to $e_1$ and $e_2$ and imposing the Neumann condition at $v$ yields a quantum graph $(\dot E,Q)$.  The same symbol $Q$ is used for the operator because the Neumann condition guarantees continuity of value and derivative across $v$; thus $(E,Q)$ and $(\dot E,Q)$ are essentially identical quantum graphs.

Denote the transfer matrices for $-d^2/dx^2 + q(x)$ on $[0,L]$, on $[0,\ell]$, and on $[\ell,L]$ by
\begin{equation}
  T(\lambda)=\mat{1.2}{c(\lambda)}{s(\lambda)}{c'(\lambda)}{s'(\lambda)},
  \quad
  T_1(\lambda)=\mat{1.2}{c_1(\lambda)}{s_1(\lambda)}{c_1'(\lambda)}{s_1'(\lambda)},
  \quad
  T_2(\lambda)=\mat{1.2}{c_2(\lambda)}{s_2(\lambda)}{c_2'(\lambda)}{s_2'(\lambda)}.
\end{equation}
Considering $(\dot E,Q)$ as one period of a $d$-periodic disconnected graph, the dispersion function $D(z,\lambda)$ is a meromorphic function $\dhRR(\lambda)$ of $\lambda$ alone, as its discrete reduction~$\hat Q(z,\lambda)$ is independent of $z$,
\begin{equation}\label{hRR}
\dhRR(\lambda) \;=\; \det
  \renewcommand{\arraystretch}{1.3}
\left[\!
\begin{array}{ccc}
  -\frac{c_1(\lambda)}{s_1(\lambda)}-\alpha_1 & \frac{1}{s_1(\lambda)} & 0\\
  \vspace{-2ex}\\
  \frac{1}{s_1(\lambda)} & -\frac{s_1'(\lambda)}{s_1(\lambda)}-\frac{c_2(\lambda)}{s_2(\lambda)} & \frac{1}{s_2(\lambda)}\\
  \vspace{-2ex}\\
  0 & \frac{1}{s_2(\lambda)} & -\frac{s_2'(\lambda)}{s_2(\lambda)}-\alpha_2
\end{array}
\!\right].
\end{equation}
When the homogeneous Dirichlet condition is imposed at either end (not both ends) of $[0,L]$, one obtains
\begin{equation}
  \dhDR(\lambda) \;=\; \det
\left[\!
\begin{array}{cc}
 -\frac{s_1'(\lambda)}{s_1(\lambda)}-\frac{c_2(\lambda)}{s_2(\lambda)} & \frac{1}{s_2(\lambda)}\\
  \vspace{-2ex}\\
  \frac{1}{s_2(\lambda)} & -\frac{s_2'(\lambda)}{s_2(\lambda)}-\alpha_2
\end{array}
\!\right],
\quad
  \dhRD(\lambda) \;=\; \det
\left[\!
\begin{array}{cc}
  -\frac{c_1(\lambda)}{s_1(\lambda)}-\alpha_1 & \frac{1}{s_1(\lambda)} \\
  \vspace{-2ex}\\
  \frac{1}{s_1(\lambda)} & -\frac{s_1'(\lambda)}{s_1(\lambda)}-\frac{c_2(\lambda)}{s_2(\lambda)}\end{array}
\!\right],
\end{equation}
and\, $\dhDD(\lambda) = -\frac{s_1'(\lambda)}{s_1(\lambda)}-\frac{c_2(\lambda)}{s_2(\lambda)}$ when the Dirichlet condition is imposed at both ends.
Using the relation $T = T_2 T_1$, one computes that
\begin{equation}
  \dhRR = -\frac{c'+\alpha_1 s' + \alpha_2 c + \alpha_1\alpha_2 s}{s_1s_2},
  \quad
  \dhDR = \frac{s'+\alpha_2 s}{s_1s_2},
  \quad
  \dhRD = \frac{c + \alpha_1 s}{s_1s_2},
  \quad
  \dhDD = -\frac{s}{s_1s_2}.
\end{equation}
For the un-dotted quantum graph $(E,Q)$, one obtains these same expressions except with the denominator $s_1(\lambda)s_2(\lambda)$ replaced by~$s(\lambda)$,
\begin{equation}
  \hRR = -\frac{c'+\alpha_1 s' + \alpha_2 c + \alpha_1\alpha_2 s}{s},
  \quad
  \hDR = \frac{s'+\alpha_2 s}{s},
  \quad
  \hRD = \frac{c + \alpha_1 s}{s},
  \quad
  \hDD = -\frac{s}{s}.
\end{equation}
Observe that, given $\lambda_0$, one can guarantee that $s_1(\lambda_0)s_2(\lambda_0)\not=0$ by choosing the point $\ell$ not to be a root of any Dirichlet eigenfunction of $-d^2/dx^2+q(x)$ for $\lambda_0$ on $[0,L]$.

These calculations show that the numerators in the expressions above contain the essential spectral information.  In fact this is true of periodic quantum graphs in general.  To go from the dispersion function for a quantum graph $(\Gamma,A)$ to the dispersion function for a dotted version $(\dot\Gamma,A)$, one simply multiplies by a factor of the form $s(\lambda)/(s_1(\lambda)s_2(\lambda))$ for each dotted edge.

\begin{proof}[Proof of Proposition~\ref{prop:poles}]
If $v_1$ and $v_2$ are not in the same $\ZZ^d$ orbit, we can assume that they both are in the vertex set $\V_0$ of the fundamental domain chosen for constructing $\hat A(z,\lambda)$, since $D(z,\lambda)$ is independent of that choice.
Denote by $\hat A(z,\lambda)$ and $\hat{\dot A}(z,\lambda)$ the discrete reductions at energy $\lambda$ of the quantum graphs $(\Gamma,A)$ and~$(\dot\Gamma,A)$.  Index the rows and columns of $\hat A(z,\lambda)$ so that the first two correspond to $v_1$ and $v_2$; then augment it with a $0^\text{th}$ column and row consisting of a $1$ in the $(0,0)$ entry and zeroes elsewhere.  Call this matrix $\hat{\tilde A}(z,\lambda)$.

The matrix $\hat{\tilde A}(z,\lambda)$ has the block form
\begin{equation}
  \renewcommand{\arraystretch}{1.3}
\left[
\begin{array}{c|c}
  \Sigma + A & B \\\hline
                 C & D
\end{array}
\right],
\end{equation}
in which
\begin{equation}
  \Sigma =
  \renewcommand{\arraystretch}{1.1}
\left[
\begin{array}{ccc}
  1 & 0 & 0 \\
  0 & -cs^{-1} & s^{-1} \\
  0 & s^{-1} & -s's^{-1}
\end{array}
\right],
\end{equation}
$A$ and $B$ have all zeroes in the first row, and $A$ and $C$ have all zeroes in the first column.  The variable $z$ does not appear in~$\Sigma$ because $v$ and $w$ are in the same 
The matrix $\hat{\dot A}(z,\lambda)$ is obtained by replacing $\Sigma$ by $\dot\Sigma$, and $\dot\Sigma$ is obtained from $\hRR(\lambda)$ (eq.\,\ref{hRR}) with $\alpha_1\!=\!\alpha_2\!=\!0$, by switching the first two rows and the first two columns (that is, switching the order of the vertices from $(v_1,v,v_2)$ to $(v,v_1,v_2)$), to obtain
\begin{equation}
  \dot\Sigma =
  \renewcommand{\arraystretch}{1.1}
\left[
\begin{array}{ccc}
  -s\,s_1^{-1}s_2^{-1} & s_1^{-1} & s_2^{-1} \\
  s_1^{-1} & -c_1s_1^{-1} & 0 \\
  s_2^{-1} & 0 & -s_2's_2^{-1}
\end{array}
\right],
\end{equation}
where the relation $s=s_1c_2+s_1's_2$ is used in the upper left entry.

The $3\!\times\!3$ matrix $K = A-BD^{-1}C$ has all zeroes in its first row and first column.  A computation using the relation $T=T_2T_1$ yields the key relation
\begin{equation}\label{ssdsd}
  s_1(\lambda)s_2(\lambda)\,\det(\dot\Sigma+K) \;=\; s(\lambda)\, \det(\Sigma+K),
\end{equation}
which holds for any matrix $K$ whose first column and row vanishes.
Using this together with
\begin{equation}
  \det \hat{\dot A} = \det D\,\det(\dot\Sigma + K),
  \qquad
 \det \hat A = \det \hat{\tilde A} = \det D \det(\Sigma+K)
\end{equation}
yields the statement of the theorem.

If $v_2=gv_1$ for some $g\in\ZZ^d$, the process above remains the same, except that
\begin{equation}
  \Sigma \,=\,
  \renewcommand{\arraystretch}{1.1}
\frac{1}{s}
\left[
\begin{array}{cc}
  1 & 0 \\
  0 & -c-s'+z^g+z^{-g}
\end{array}
\right],
\qquad
  \dot\Sigma \,=\,
  \renewcommand{\arraystretch}{1.1}
\frac{1}{s_1s_2}
\left[
\begin{array}{cc}
  -s & s_2+z^gs_1 \\
  s_2+z^{-g}s_1 & -c_1s_2-s_2's_1
\end{array}
\right],
\end{equation}
and $K$ is a $2\!\times\!2$ matrix with its only nonzero entry being the lower right.  In this case, one obtains (\ref{ssdsd}) with an extra minus sign on one side.
\end{proof}

%%%%%%%%%%%%%%%%%%%%%%%%%%%%%%%%%%%%%%%%%%%%%%%%%

\bigskip
\noindent{\bfseries Acknowledgement.}
This material is based upon work supported by the National Science Foundation under Grant No. DMS-1814902.

\end{document}